\documentclass[pra,aps,amsmath, amssymb, twocolumn, superscriptaddress,longbibliography]{revtex4-2}

\usepackage{graphicx}
\usepackage[colorlinks=true, linkcolor=red, urlcolor=blue]{hyperref}
\usepackage{physics}
\usepackage{color}
\usepackage{enumerate}
\usepackage{bm}
\usepackage[normalem]{ulem}
\usepackage{comment}
\usepackage{dsfont}
\usepackage{multirow}
\usepackage{tikz}
\usetikzlibrary{arrows}
\usetikzlibrary{backgrounds}

\usepackage[caption=false]{subfig}

\usepackage{braket}
\usepackage{amsmath}

\usepackage{amsthm}
\usepackage{amsfonts, amssymb}
\usepackage{bbm}

\newtheorem{lemma}{Lemma}
\newtheorem{theorem}{Theorem}
\newtheorem{remark}{Remark}

\usepackage{mathtools}
\usepackage{mathrsfs}
\usepackage{orcidlink}

\newcommand{\ergo}{\mathcal{E}}

\newcommand{\en}{\mathfrak{E}}

\newcommand{\dstate}{\hat{\rho}}

\newcommand{\ham}{\hat{H}}

\newcommand{\adc}{\Phi_{\gamma}}
\newcommand{\gadc}{\Phi_{\gamma,\eta}}
\newcommand{\multiergo}{\mathcal{J}_{\mathcal{E}}}
\newcommand{\ADCn}{\Phi^{\otimes n}_\gamma}

\newcommand{\fixedergo}{\bar{\mathcal{E}}}

\begin{document}
	
	\title{Quantum work extraction efficiency for noisy quantum batteries:\\  the role of coherence} 
 
\author{Salvatore Tirone\orcidlink{0000-0002-4880-4329}}
\email{s.tirone@uva.nl}
\affiliation{Scuola Normale Superiore, I-56126 Pisa, Italy}
\affiliation{QuSoft, Science Park 123, 1098 XG Amsterdam, the Netherlands}
\affiliation{Korteweg--de Vries Institute for Mathematics, University of Amsterdam, Science Park 105-107, 1098 XG Amsterdam, the Netherlands}
 
	\author{Raffaele Salvia}
\email{raffaele.salvia@sns.it}
	\affiliation{Scuola Normale Superiore, I-56126 Pisa, Italy}
    
	\author{Stefano Chessa}
 \email{schessa@illinois.edu}
	\affiliation{NEST, Scuola Normale Superiore and Istituto Nanoscienze-CNR, I-56126 Pisa, Italy}
	\affiliation{Electrical and Computer Engineering, University of Illinois Urbana-Champaign, Urbana, Illinois, 61801, USA}
    
	\author{Vittorio Giovannetti}
	\affiliation{NEST, Scuola Normale Superiore and Istituto Nanoscienze-CNR, I-56126 Pisa, Italy}
	
	\begin{abstract}		
Quantum work capacitances and maximal asymptotic work/energy ratios
are figures of merit characterizing the robustness against noise of work extraction processes in quantum batteries formed by collections of quantum systems. In this paper we establish a direct connection between these functionals and, exploiting this result, we analyze different types of noise models mimicking self-discharging (amplitude damping channels), thermalization (generalized amplitude damping channels) and dephasing effects (simultaneous phase-amplitude damping channels). In this context we show that input quantum coherence can significantly improve the storage performance of noisy quantum batteries and that the maximum output ergotropy is not always achieved by the maximum available input energy. 
	\end{abstract}
	
	\maketitle

\section{Introduction}

The use of quantum effects have been proposed as useful methods to speed up 
charging processes of batteries \cite{Alicki2013, Campaioli2018, Andolina2019, Farina2019, Rossini2019, rossini2019quantum, JuliFarr2020, PhysRevA.107.032218, erdman2022reinforcement}. To assess the stability of these proposals it is crucial to analyze how such schemes perform in the presence of environmental noise, as this represents a more realistic scenario, as it has been proved by the first experimental implementations \cite{Quach2022, Hu_2022, exp_battery, exp_battery1, exp_battery2, exp_battery3},  and can have a significant impact on the efficiency of energy recovery. For noiseless models capacities for Quantum Batteries (QBs) have been proposed \cite{def_ergo, JuliFarr2020, battery_cap2023} while, from a resource theoretical point of view, the thermodynamic capacity (in the sense of process simulability) of quantum channels has been defined \cite{Thermodyn_Cap2019}. In the presence of noise, in terms of energetic manipulation, a few results have been obtained in specific frameworks, see e.g.  \cite{Carrega_2020, Bai_2020, Tabesh_2020, Ghosh_2021, Santos_2021, Zakavati_2021, Landi_2021, Morrone_2022, Sen_2023,Liu2019, PhysRevE.100.032107, PhysRevApplied.14.024092, Bai_2020, PhysRevE.101.062114, PhysRevResearch.2.013095,  liu2021boosting, PhysRevE.103.042118, PhysRevE.105.054115}. In an effort to generalize and facilitate the comparison between different physical platforms, two types of universal figures of merit, i.e. the quantum work capacitances and  the maximal asymptotic work/energy ratios (MAWERs), have been introduced in~\cite{quantumworkcapacitances}.These quantities gauge different aspects of the efficiency of work extraction from noisy quantum battery systems  providing a parallel point of view over the concept of efficiency of other quantum thermodynamic objects \cite{Guryanova2016,PhysRevE.96.012128,PhysRevE.97.012129}, such as e.g. quantum heat engines, which model alternative thermodynamic operational protocols \cite{gemmer_quantum_2009, deffner2019quantum, Goold2016}. Here, the aforementioned new figures of merit are useful in scenarios where one has at disposal large collections of identical energy storing quantum devices: quantum cells or q-cells in brief.
Specifically, quantum work capacitances target the maximum work that can be recovered per q-cell 
assuming that on average each cell allocates a fixed portion of the total energy $E$ initially stored in the overall battery.  For MAWERs, on the contrary,
the work extraction efficiency  is gauged treating  the q-cells  as a free resource, assuming their number $n$ to be much larger than the minimal amount that could host the full input energy $E$.   
The evaluations of quantum work capacitances and of MAWERs rely on complex and resource-constrained optimization problems  
similar to those one faces in quantum communication~\cite{nielsen_chuang_2010, HOLEGIOV, ADV_Q_COMM, HAYASHIBOOK, WILDEBOOK, SURVEY, HOLEVOBOOK}, where encoding/decoding strategies must be designed to mitigate the effects of dissipation and decoherence -- see Fig.~\ref{fig:noise_bat}.
Closed formulas that enable one to explicitly compute the value of some of the 
work capacitances have been recently presented in Ref. \cite{tirone2023work}. Here we extend these findings to the MAWERs, proving that the latter can be expressed as a special limit of the associated quantum work capacitance functionals.
Using these results we next investigate 
several models of noisy batteries describing a wide range of deterioration effects that are physically relevant,  including self-discharging, thermalization, and dephasing.
In this context we find that the work that can be extracted from the QB is not always a monotonic  function of the energy that was initially stored in the system. As a consequence we show that, for some noise models, an incompletely charged battery may perform better than a fully charged one.
Additionally, we observe that the presence of quantum coherence strengthens the resistance of QBs against environmental noise. Specifically, when it comes to self-discharging and thermal noise, we show that coherent input states outperform all conceivable classical (i.e. incoherent) initial states of the battery. This result strengthens the evidences built up in the literature in recent years, both in the noisy and the noiseless settings \cite{Pintos_2020, Caravelli2021energystorage, Shi_2022, Mayo_2022, Yu_2023}, that quantum coherence can provide an advantage in the quantum setting for energy manipulation tasks and proves to be a useful resource as happens for information processing applications \cite{Streltsov_2017, Saxena_2020, Kamin2020, Selby_2020}.  

The manuscript is organized as follows:
in Sec.~\ref{sec:definizioni} we introduce the problem and derive a general relation which, irrespective of the noise model, allows us to directly connect MAWERs to quantum work capacitances;
in Sec.~\ref{sec:canali} we study the values of these figures of merit for some specific noise models for quantum battery systems made of collections of two-level quantum cells;
conclusions are drawn in Sec.~\ref{sec:discuss}.

\bigskip

\begin{figure*}[t!]
\centering
\includegraphics[width=0.8\linewidth]{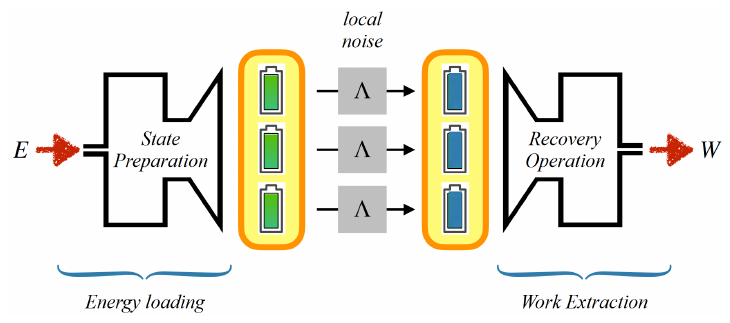} 
\caption{Schematic representation of a quantum work extraction protocol operating on a  noisy quantum battery composedby $n$ q-cells. Before the action of the noise (grey elements of the figure), some initial energy $E$ is loaded into the system via a state preparation  procedure (leftmost element of the figure) which mimics the encoding stage of quantum communication settings. After the noise action, some work $W$ is extracted form the system via recovery operations (rightmost element) playing the role of signal decoding.}
\label{fig:noise_bat}
\end{figure*}

\section{The model}
\label{sec:definizioni}
 The work extraction procedures we analyze are graphically depicted in Fig.~\ref{fig:noise_bat}. Similarly to~\cite{quantumworkcapacitances, tirone2023work}, they focus on QB models formed by a collection of $n$ identical and independent (not-interacting) elements (quantum cells or q-cells in brief), each capable to store energy in the internal degrees of freedom associated with the associated local Hamiltonians $\hat{h}_1$, $\hat{h}_2$, $\cdots$, $\hat{h}_n$ (without loss of generality hereafter we will assume the operator $\hat{h}$ to have zero-ground state eigenvalue). 

 In this setting an amount $E\geq 0$ of energy is loaded into the system by selecting an input density matrix $\dstate^{(n)}$ of the q-cells from the set of input states of fixed mean energy $E$. Such operation is mediated by a properly designed state preparation stage (left element of the figure) which, depending on the available resources, may or may not allow for the creation of quantum correlations among the various QB elements.
We let hence the system evolve under the action of environmental noise described in terms of a completely, positive, trace preserving (CPTP) super-operator $\Lambda$~\cite{STINE, CHOI1975285, kraus1983states} that acts identically and independently on each of the q-cells, i.e.
\begin{eqnarray} 
\label{transfnoise} 
\dstate^{(n)} \mapsto  \Lambda^{\otimes n}( \dstate^{(n)})\;. 
\end{eqnarray} 
 The negative impact of the transformation~(\ref{transfnoise}) is then evaluated by 
 using quantum work capacitance and MAWER functionals to compute the maximum amount of work that one can extract from  the deteriorated 
state of the QB, optimizing the performance over an assigned set of recovery operations (right element of the figure), possibly affected by resource  constraints.  

 \subsection{Quantum Work Capacitances} \label{sec:QWC} 
The first example of quantum work capacitance introduced in~\cite{quantumworkcapacitances} assumes the ergotropy functional~\cite{def_ergo} as
 measure of the extractable work that one can get from a quantum state.
For a $d$-dimensional system characterized by an Hamiltonian $\ham$, the ergotropy of an input state $\dstate$ can be expressed as  
\begin{eqnarray}\label{ergoeeeDEFinvariance} 
\ergo(\dstate;\ham)  := \max_{\hat{U} \in {\mathbb U}(d)}\Big\{  \en(\dstate;\ham) 
-\en(\hat{U} \dstate \hat{U}^\dag;\ham)\Big\}   \;,
\end{eqnarray}
where $\en(\dstate ; \ham) := \Tr[\dstate\ham]$ is the average energy of a quantum state $\dstate$ and $\mathbb{U}(d)$ is the $d$-dimensional representation of the unitary group (see App.~\ref{sec.rev}  for a concise review on the main properties of ergotropy). 
 Accordingly the ergotropic capacitance of the model corresponds to  
 \begin{eqnarray} 
C_{\ergo}\left( \Lambda, \mathfrak{e} \right) := \lim_{n\rightarrow \infty} 
\frac{\ergo^{(n)}(\Lambda;E=n\mathfrak{e})}{n} \;, \label{cergo} 
\end{eqnarray} 
where, given $\ham^{(n)}:=\hat{h}_1 + \cdots + \hat{h}_n$  the QB Hamiltonian, we set~\cite{PhysRevLett.127.210601} 
  \begin{eqnarray} \label{fixedergo}
 \fixedergo^{(n)}(\Lambda;E)&:=& \max_{\dstate^{(n)} \in \overline{\mathfrak{S}}^{(n)}_{E}}
 \ergo(\Lambda^{\otimes n}(\dstate^{(n)});\ham^{(n)}) \;,  \\ \label{dergo}
 \ergo^{(n)}(\Lambda;E)&:=& \max_{0 \leq E'\leq E} \fixedergo^{(n)}(\Lambda;E')  \nonumber \\
 &=&  \max_{\dstate^{(n)} \in {\mathfrak{S}}^{(n)}_{E}}
 \ergo(\Lambda^{\otimes n}(\dstate^{(n)});\ham^{(n)})\;,
 \end{eqnarray} 
the maximization for~(\ref{fixedergo}) being performed over all possible input states 
$\dstate^{(n)}\in\overline{\mathfrak{S}}^{(n)}_{E}$ of the $n$ cells whose initial energy is exactly equal to a fixed value $E$, and for~(\ref{dergo}) over all possible input states $\dstate^{(n)}\in{\mathfrak{S}}^{(n)}_{E}$ of the $n$ cells whose initial energy is  not larger than to $E$. 
Restricted versions of $C_{\ergo}\left( \Lambda; \mathfrak{e} \right)$ can be obtained by 
constraining the allowed operations one can perform on the battery before and/or after the action of the noise.
For instance, assuming the maximization in Eq.~(\ref{fixedergo}) only to run on separable input states of the q-cells will lead us to the separable-input ergotropic capacitance 
$C_{\rm sep} \left( \Lambda; \mathfrak{e} \right)$ which represents the asymptotic work we can extract per q-cell in the absence of initial entanglement between cells.
 Similarly, restricting the optimization in Eq.~(\ref{ergoeeeDEFinvariance}) to include only unitary operations acting locally on the q-cells (i.e. replacing $\ergo(\dstate;\ham)$ with the local ergotropy~(\ref{ergoeeeDEFinvarianceloc}) \cite{bipartiteworkextraction}) leads to the local ergotropic capacitance 
 $C_{\rm loc} \left( \Lambda; \mathfrak{e} \right)$ that gives the maximum work extraction rate attainable with only local resources at the decoding stage.
Assuming finally the optimizations to be restricted to both separable states and to local unitary operations one can define 
the separable-input, local ergotropic capacitance 
 $C_{\rm loc, sep} \left( \Lambda; \mathfrak{e} \right)$. 
 While clearly $C_{\cal E} \left( \Lambda; \mathfrak{e} \right)$ and 
 $C_{\rm loc, sep} \left( \Lambda; \mathfrak{e} \right)$ represent respectively the largest and lowest work capacitances of the model, an absolute ordering can be established between all  these quantities, i.e. 
\begin{eqnarray} \label{ordering} 
C_{\cal E} \left( \Lambda; \mathfrak{e} \right) \geq 
C_{\rm sep} \left( \Lambda; \mathfrak{e} \right) \geq 
C_{\rm loc} \left( \Lambda; \mathfrak{e} \right)= C_{\rm loc, sep} \left( \Lambda; \mathfrak{e} \right)\;. \end{eqnarray}
The last two relations follow from arguments in Ref. \cite{tirone2023work}, where a simplified closed formula for $C_{\rm loc} \left( \Lambda; \mathfrak{e} \right)$ and $C_{\rm loc, sep} \left( \Lambda, \mathfrak{e} \right)$ has been derived. Specifically it is proven that, irrespective of the specifics of the QB model, the encoding/decoding optimizations that express
the local and seperable-input local capacitances can be solved, leading to the identity 
 \begin{eqnarray} \label{exact} 
C_{\rm loc} \left( \Lambda; \mathfrak{e} \right)&=&
C_{\rm loc, sep} \left( \Lambda, \mathfrak{e} \right)= \chi(\Lambda; \mathfrak{e}) \;, \end{eqnarray}
where we defined the quantity $\chi(\Lambda;\mathfrak{e})$ as: 
\begin{eqnarray} \label{defchi} 
\chi(\Lambda; \mathfrak{e}) &:=& 
\sup_{\{ p_j, \mathfrak{e}_j\}} \sum_j p_j\;  \ergo^{(1)}(\Lambda;\mathfrak{e}_j)\;, 
\end{eqnarray} 
the supremum being taken over all the distributions $\{ p_j, \mathfrak{e}_j \}$ of input q-cell energy $\mathfrak{e}_j\in [0,1]$ 
fulfilling the constraint
\begin{eqnarray}\label{energyconstraint} 
\sum_j p_j  \mathfrak{e}_j\leq	 \mathfrak{e} \;. 
\end{eqnarray} 
Unfortunately, at present, no analog of Eq.~(\ref{exact}) are known for $C_{\cal E} \left( \Lambda, \mathfrak{e} \right)$ and
$C_{\rm sep} \left( \Lambda, \mathfrak{e} \right)$. Yet a nontrivial lower bound for $C_{\rm sep} \left( \Lambda, \mathfrak{e} \right)$
can be obtained by replacing in~(\ref{defchi}) the ergotropy with its $n$-fold regularized counterpart, i.e. the  total-ergotropy functional~(\ref{totergoeeeDEFtotal})~\cite{Alicki2013,Niedenzu2019}, i.e. 
 \begin{eqnarray} \label{lower} 
C_{\rm sep} \left( \Lambda; \mathfrak{e} \right)&\geq& \chi_{\rm tot}(\Lambda; \mathfrak{e}) \;, 
\end{eqnarray} 
with 
\begin{eqnarray} 
\chi_{\rm tot}(\Lambda; \mathfrak{e}) &:=&  \label{defchitot} 
\sup_{\{ p_j, \mathfrak{e}_j\}} \sum_j p_j\;  \ergo_{\rm tot}^{(1)}(\Lambda;\mathfrak{e}_j)\;. 
\end{eqnarray} 
Furthermore, according to Corollary 2 of Ref. \cite{tirone2023work}, the gap in  Eq.~(\ref{defchitot}) closes at least for all those noise models $\Lambda$ which, given $\mathfrak{e}\in [0,1]$, admit a single-site state $\hat{\sigma}_{\mathfrak{e}}$ of mean energy not larger than $\mathfrak{e}$ such that  
\begin{equation}\label{step1} 
\ergo((\Lambda(\hat{\sigma}_{\mathfrak{e}})) ^{\otimes n};\ham^{(n)}) \geq \ergo(\Lambda^{\otimes n}(|\Psi_{\rm fact}^{(n)}\rangle\!\langle \Psi_{\rm fact}^{(n)}|);\ham^{(n)})\;, 
\end{equation} 
for all 
\begin{eqnarray}\label{defPSIfac} 
|\Psi_{\rm fact}^{(n)}\rangle := |\psi_1\rangle\otimes |\psi_2\rangle \otimes \cdots \otimes |\psi_n\rangle\;, \end{eqnarray} 
pure, factorized state of the q-cells 
with mean energy smaller than or equal to $E=n\mathfrak{e}$.

\subsection{Maximal asymptotic Work/Energy ratios} \label{sec:MAWER} 

While the quantum work capacitances are defined for a finite ratio of the input energy $E$ with respect to the number $n$ of q-cells (i.e., $E/n$ is fixed), MAWERs are computed instead assuming $n$ to be a free resource, independent of $E$. 
Formally, computing the MAWER corresponding to 
$C_{\ergo}\left( \Lambda, \mathfrak{e} \right)$  
accounts to solving the following limit: 
\begin{eqnarray} \label{MAWERDEF} 
\multiergo(\Lambda) := \limsup_{E\rightarrow\infty}\left(\sup_{n\geq \lceil E/\mathfrak{e}_{\max} \rceil}\frac{   \ergo^{(n)}(\Lambda;E)}{E}\right) \;,
\end{eqnarray}
with $\mathfrak{e}_{\max}$ being the maximum eigenvalue of the single-site Hamiltonian $\hat{h}$. The definition in Eq. (\ref{MAWERDEF}) can be proved to be equivalent to the following formula:
\begin{equation} \label{eq:truemaw}
\mathcal{J}_{\ergo}(\Lambda) = \lim_{E\to\infty}\lim_{n\to\infty}\frac{\ergo^{(n)}(\Lambda;E)}{E} \; .
\end{equation}
This latter characterization gives a physical meaning to the quantity above: it shows that in principle this measure of efficiency can be practically achieved by using a large number of q-cells. The proof of Eq. (\ref{eq:truemaw}) can be found in appendix \ref{app:physmaw}.
Analogously to what was done for the quantum work capacitances, also for the MAWER one can introduce a hierarchy of constrained versions of such quantity differing on the resources dedicated to the work extraction process. Specifically, 
replacing $\ergo^{(n)}(\Lambda;E)$ appearing in the r.h.s. of Eq.~(\ref{MAWERDEF}) with 
$\ergo_{\rm loc}^{(n)}(\Lambda;E)$, $\ergo_{\rm sep}^{(n)}(\Lambda;E)$ and $\ergo_{\rm loc, sep}^{(n)}(\Lambda;E)$, we can speak of 
local MAWER ${\cal J}_{\rm loc}(\Lambda)$, separable-input MAWER ${\cal J}_{\rm sep}(\Lambda)$, and separable-input local MAWER ${\cal J}_{\rm loc, sep}(\Lambda)$, respectively. For all of the constrained versions of the MAWER defined above holds a characterization analogous to Eq. (\ref{eq:truemaw}).

We now prove an important relation allowing us to establish a direct connection between each one of such functionals and their associated quantum work-capacitance:
\begin{theorem}\label{theo1} For any given CPTP map $\Lambda$, the following identity holds
\begin{eqnarray}
{\cal J}_{\diamond}(\Lambda) = \lim_{\mathfrak{e}\rightarrow 0} \frac{C_{\diamond}(\Lambda;\mathfrak{e})}{\mathfrak{e}}\;, \label{JvsQW} 
\end{eqnarray} 
where hereafter we will use the symbol $\diamond$ as a placeholder variable taking values on the set $\{ {\cal E}; {\rm sep}; {\rm loc}; {\rm loc, sep}\}$.
\end{theorem} 
\begin{proof}
Let us start by noticing that ${\cal J}_{\diamond}(\Lambda)$ can be always lower bounded as
 \begin{equation} \label{MAWERDEFineqtrivialnew1} 
{\cal J}_{\diamond}(\Lambda) \geq \sup_{\mathfrak{e}\in (0,\mathfrak{e}_{\max}]} \frac{C_{\diamond}(\Lambda;\mathfrak{e})}{\mathfrak{e}}\;. 
\end{equation}
Indeed since for $\mathfrak{e}\in (0,\mathfrak{e}_{\max}]$ one has 
 $\lceil E/\mathfrak{e}\rceil \geq \lceil E/\mathfrak{e}_{\max}\rceil$, we can write 
 \begin{eqnarray} \label{MAWERDEFineq} 
{\cal J}_{\diamond}(\Lambda) &\geq&  \limsup_{E\rightarrow\infty}\frac{   \ergo_{\diamond}^{(\lceil E/\mathfrak{e}\rceil)}(\Lambda;E)}{E} \nonumber \\
&\geq& \limsup_{E\rightarrow\infty}\frac{\ergo_{\diamond}^{(\lceil E/\mathfrak{e}\rceil-1)}(\Lambda;(\lceil E/\mathfrak{e}\rceil -1)\mathfrak{e})}{E} \nonumber \\ 
&=&  \limsup_{E\rightarrow\infty}\nonumber
\left(\tfrac{\lceil E/\mathfrak{e}\rceil -1}{E} \right)
\frac{\ergo_{\diamond}^{(\lceil E/\mathfrak{e}\rceil-1)}(\Lambda;(\lceil E/\mathfrak{e}\rceil -1)\mathfrak{e})}{\lceil E/\mathfrak{e}\rceil -1}
\\&=& \frac{C_{\diamond}(\Lambda;\mathfrak{e})}{\mathfrak{e}}\;,
\end{eqnarray}
where the second inequality follows from 
$E\geq (\lceil E/\mathfrak{e}\rceil -1)\mathfrak{e}$ and from the fact that 
$\ergo_{\diamond}^{(n)}(\Lambda;E)$ is monotonically increasing 
w.r.t. to $E$ and $n$~\cite{quantumworkcapacitances}; the final identity is obtained by direct computation of the $E\rightarrow \infty$ limit. 
Evaluated for $\mathfrak{e}\rightarrow 0$, Eq.~(\ref{MAWERDEFineqtrivialnew1}) implies in particular that the r.h.s. of Eq.~(\ref{JvsQW}) is also a lower bound for ${\cal J}_{\diamond}(\Lambda)$, i.e.
\begin{eqnarray}{\cal J}_{\diamond}(\Lambda) \geq \lim_{\mathfrak{e}\rightarrow 0} \frac{C_{\diamond}(\Lambda;\mathfrak{e})}{\mathfrak{e}} \; . 
\end{eqnarray} 
Accordingly to prove the thesis we only need to show that also the reverse inequality holds true.
For this purpose observe that, for $E\geq 0$ and $n\geq \lceil E/ \mathfrak{e}_{\max}\rceil$ integer, we can write
\begin{eqnarray}
\ergo_{\diamond}^{(n)}(\Lambda;E) =\ergo_{\diamond}^{(n)}(\Lambda;n E/n)\leq
n C_{\diamond}(\Lambda; E/n) \;, 
\end{eqnarray} 
where in the final passage we used the fact that 
$\ergo_{\diamond}^{(n)}(\Lambda;n \mathfrak{e}) \leq n C_{\diamond}(\Lambda; \mathfrak{e})$
for all $\mathfrak{e}\in [0,\mathfrak{e}_{\max}]$ and $n$ integer. 
Since $\ergo^{(n)}(\Lambda;E)$ is monotonically non-decreasing in $n$ for any $E$ as shown in \cite{quantumworkcapacitances} we obtain
 \begin{eqnarray} \mathcal{J}_{\diamond}(\Lambda) &\leq&  \limsup_{E\rightarrow\infty}\left(\lim_{n\rightarrow \infty}\frac{   C_{\diamond}(\Lambda; E/n)}{E/n}\right) \nonumber \\
\label{MAWERDEFgeq}  
 &=&  \lim_{\mathfrak{e}\rightarrow 0} \frac{C_{\diamond}(\Lambda;\mathfrak{e})}{\mathfrak{e}}\;.
 \end{eqnarray}
\end{proof} 

\begin{remark} For channels $\Lambda$ such that $C_{\diamond}(\Lambda,\mathfrak{e})$ admits a nonzero value for $\mathfrak{e}=0$, Eq.~(\ref{JvsQW}) implies that the associated MAWER diverges, i.e. 
\begin{eqnarray}
C_{\diamond}(\Lambda;0)\neq 0 \qquad \Longrightarrow \qquad 
 {\cal J}_{\diamond}(\Lambda) = \infty \;. \label{JvsQWinfity} 
\end{eqnarray} 
For the ergotropic and separable-input ergotropic capacitances (i.e. for $\diamond = {\cal E}$ and ${\rm sep}$) this happens when $\Lambda$ maps
the ground state $|0\rangle$ into an output state which is not completely passive.
For the local ergotropic capacitances (i.e. $\diamond = {\rm loc}$ and ${\rm loc, sep}$), Eq.~(\ref{JvsQWinfity}) occurs instead for all those maps for which $\Lambda(|0\rangle\!\langle 0|)$ is a non passive configuration.
\end{remark} 
In~\cite{quantumworkcapacitances} is proven that the depolarizing channel is an example of channel whose MAWER diverges.

\begin{remark} \label{remark2} For channels $\Lambda$ such that the function $C_{\diamond}(\Lambda,\mathfrak{e})$ is differentiable in $\mathfrak{e}=0$ and there assumes zero value, Eq.~(\ref{JvsQW}) can be expressed as
\begin{eqnarray}
{\cal J}_{\diamond}(\Lambda) = \left. \frac{\partial C_{\diamond}(\Lambda,\mathfrak{e})}{\partial \mathfrak{e}}\right|_{\mathfrak{e}=0} \;. \label{JvsQWderi} 
\end{eqnarray}
This also implies that the work capacitance admits the following linear expansion for $\mathfrak{e}\ll \mathfrak{e}_{\max}$, i.e. 
\begin{eqnarray}
{C}_{\diamond}(\Lambda,\mathfrak{e}) = {\cal J}_{\diamond}(\Lambda) \mathfrak{e} +
O(({\mathfrak{e}}/{\mathfrak{e}_{\max}})^2)\;. \label{JvsQWderiinv} 
\end{eqnarray}
\end{remark} 
\begin{remark} In view of the identity (\ref{exact}), Theorem~\ref{theo1} implies that 
the local MAWER always coincides with the separable-input local MAWER. In particular 
they can be expressed as 
\begin{eqnarray} \label{exactMAW} 
{\cal J}_{\rm loc} \left( \Lambda \right)&=&
{\cal J}_{\rm loc, sep} \left( \Lambda \right)= \chi'(\Lambda)\;,
\end{eqnarray}
with 
\begin{eqnarray}\chi'(\Lambda)&:=& 
 \lim_{\mathfrak{e}\rightarrow 0} \frac{\chi(\Lambda,\mathfrak{e})}{\mathfrak{e}}
=\left.\frac{\partial{\chi (\Lambda,  \mathfrak{e})}}{\partial  \mathfrak{e}}\right|_{\mathfrak{e} =0}\;,  \label{simple} 
\end{eqnarray} 
the second identity holding whenever the function $\chi (\Lambda,  \mathfrak{e})$ of Eq.~(\ref{defchi}) is differentiable in $\mathfrak{e}=0$ and fulfils the condition~$\chi (\Lambda, 0)=0$.
Similarly, from Eq.~(\ref{lower}) we can derive a lower bound for the separable-input MAWER 
 \begin{eqnarray} \label{ineqtot} 
{\cal J}_{\rm sep}(\Lambda)\geq \chi'_{\rm tot}(\Lambda)\;,
\end{eqnarray} 
with 
\begin{eqnarray}
\chi'_{\rm tot}(\Lambda)&:=& \lim_{\mathfrak{e}\rightarrow 0} \frac{\chi_{\rm tot} (\Lambda,\mathfrak{e})}{\mathfrak{e}}
=\left.\frac{\partial{\chi_{\rm tot} (\Lambda,  \mathfrak{e})}}{\partial  \mathfrak{e}}\right|_{\mathfrak{e} =0}\;, 
\end{eqnarray} 
the second identity holding whenever $\chi_{\rm tot} (\Lambda,  \mathfrak{e})$ of~(\ref{defchitot}) is differentiable in $\mathfrak{e}=0$ and fulfils the condition~$\chi_{\rm tot}  (\Lambda, 0)=0$.
\end{remark}

\section{Multi-Qubit noisy battery models}
\label{sec:canali}
In this section we focus on QB models made of identical two-levels (qubit) q-cells characterized by local Hamiltonians $\hat{h}$ whose maximum eigenvalue $\mathfrak{e}_{\max}$ is, without loss of generality, fixed equal to $1$, i.e. 
\begin{eqnarray}
\hat{h} = \ket{1}\!\!\bra{1}\;,
\end{eqnarray}
 ($|0\rangle$ being the zero energy ground state).  For such setting we evaluate
 the work capacitances and the MAWERs of three noise models:
 \begin{itemize}
 \item  Sec. \ref{sec:adc}: the Amplitude Damping Channel (ADC) $\adc$~\cite{nielsen_chuang_2010}, a CPTP map that can be used to model self-discharging effects (i.e. the decay of population from the excited level $|1\rangle$ 
 to the ground state $|0\rangle$) in systems interacting with a zero temperature or pure state environment. It can be applied for instance to model the $T_1$ decoherence time associated to, among the prevailing platforms in quantum technologies, superconducting qubits \cite{papic2023error}, solid states qubits \cite{chirolli2008}, ion trap qubits \cite{ion_ADC} and it has been shown to be a relevant process to be addressed in experimental realizations of quantum batteries \cite{Quach2022, Hu_2022}.
 \item Sec. \ref{sec:GADC}: the Generalized Amplitude Damping Channel (GADC)  $\gadc$ \cite{qubitGADC}, that describes  partial thermalization  of the q-cells placed in contact with an external bath in a thermal state. This thermalization model can describe the behaviour of superconducting qubits \cite{papic2023error}, solid state qubits \cite{chirolli2008}, trapped atoms qubits \cite{Myatt2000, atom_GADC} in a thermal environment and, specifically, superconducting-based instances of experimental quantum batteries \cite{Hu_2022}.
 \item Sec. \ref{sec:metaphysics}: we analyze the channel resulting from the composition of ADCs with extra dephasing noise (an analysis on the dephasing channel alone was performed in~\cite{quantumworkcapacitances}). The simple model of dephasing characterizes the decoherence process in a qubit and is ubiquitous across all practical quantum technologies architectures. Specifically, the simultaneous presence of dephasing and amplitude damping effectively describes the behaviour of superconducting qubits \cite{papic2023error} and it's proved to occur in experimental quantum batteries \cite{Quach2022}. 
\end{itemize}

\subsection{Amplitude Damping Channel}
\label{sec:adc}
Expressed in the energy basis $\{ |0\rangle, |1\rangle\}$, the action of the amplitude damping channel $\adc$ on a generic density matrix $\dstate$ of a qubit system is defined as \cite{GiovannettiFazio2005}:
\begin{equation}\label{eq: ADC matrix}
\adc(\dstate)=\begin{pmatrix}
\rho_{00} + \gamma \rho_{11} & \sqrt{1-\gamma}\rho_{01} \\
\sqrt{1-\gamma}\rho_{10} & (1-\gamma)\rho_{11}
\end{pmatrix} \;,
\end{equation}
with $\gamma\in [0,1]$ being the damping parameter of the model, and $\rho_{ij}$ being the matrix elements of $\dstate$.
We recall that $\adc$ commutes with the Hamiltonian evolution of the system, i.e.
\begin{eqnarray}\label{covariance} 
\adc(e^{-i\hat{h} t}\dstate e^{i\hat{h} t}) = e^{-i\hat{h} t}\adc(\dstate)e^{i\hat{h} t}\;,
\end{eqnarray}
 and that the output energy and input
energy of any state are connected via the identity 
\begin{equation} \label{eq:adcenergylim}
\en(\adc(\dstate)) = (1-\gamma)\en(\dstate) \; .
\end{equation}
Notice also that $\adc$ maps the ground state into itself, i.e. 
 \begin{eqnarray} \label{Tzero} 
  \adc(\ket{0}\!\!\bra{0})= \ket{0}\!\!\bra{0}\;. 
 \end{eqnarray} 
As proven in \cite{PhysRevLett.127.210601}, the output ergotropy $\ergo(\Lambda(\dstate); \hat{h})$ of a generic quantum 
channel $\Lambda$ is always maximized by a pure state. Therefore, to compute the single-use fixed-energy maximum output 
ergotropy $\fixedergo^{(1)}(\adc; E=\mathfrak{e})$, it is sufficient to consider the sets of pure states with energy $\mathfrak{e}\in [0,1]$, 
which can be parameterized as
\begin{eqnarray}
\ket{\psi_{\mathfrak{e}}^{(\phi)}} := \sqrt{1-{\mathfrak{e}}}\ket{0} + e^{i\phi}\sqrt{{\mathfrak{e}}}\ket{1} \;,
\label{stato_energiaE}
\end{eqnarray}
with the phase $\phi$ that can be taken equal to zero thanks to the covariance property~(\ref{covariance})~\cite{quantumworkcapacitances}.
 By direct computation one can observe that $\adc(\ket{\psi_{\mathfrak{e}}^{(\phi)}}\!\!\bra{\psi_{\mathfrak{e}}^{(\phi)}})$ admits as eigenvalues 
 \begin{equation}
	\lambda^{(\pm)}_{\gamma}({\mathfrak{e}}):= \frac{1}{2} \left(1 \pm \sqrt{1-4 {\mathfrak{e}}^2\gamma(1-\gamma)} \right)\; ,
	\end{equation}
which, using~(\ref{ergoqubit}), leads to the following formula for the the single-shot ($n=1$), fixed energy, maximum output ergotropy of the model,
\begin{equation}
\fixedergo^{(1)}(\adc;\mathfrak{e}) = (1-\gamma){\mathfrak{e}} - \frac{1}{2}\left(1- \sqrt{1-4{\mathfrak{e}}^2\gamma(1-\gamma)}\right) \; .
\label{fixedergo_adc}
\end{equation}

As shown in Fig.~\ref{fig: plot2}, for fixed $\gamma$, this is a concave and in general non-monotonic function of the energy ${\mathfrak{e}}$. In the interval $[0,1]$, it reaches its maximum value for
 	\begin{eqnarray}
	E_{\gamma} := \min\left\{ 1, 1/({2 \sqrt{ \gamma}})\right\}\;, 
	 \end{eqnarray}
	  meaning that using all the available input energy is not always the best option to get a configuration that is more resilient to the noise.  
An inspection of Eq.~(\ref{fixedergo_adc}) also reveals that the optimal choice of the input energy is given by 
\begin{eqnarray}\label{MAX4} 
{E}^{\star}_{\gamma}({\mathfrak{e}}):= \min\{ {\mathfrak{e}}, E_{\gamma}\} = \min\{ {\mathfrak{e}}, 1/(2\sqrt{\gamma})\}\;,
\end{eqnarray}  which yields
\begin{equation} \label{eq: ADC erg_11_new}
\ergo^{(1)}(\adc;\mathfrak{e})=
\fixedergo^{(1)}(\adc;{E}^{\star}_{\gamma}({\mathfrak{e}})) \;, \end{equation}
that is also concave w.r.t. ${\mathfrak{e}}$ -- see Fig.~\ref{fig: plot2}. When replaced in Eqs.~(\ref{exact}) and~(\ref{defchi}) we can thus arrive to the following expression for the local ergotropic capacitances, 
\begin{equation} \label{defchiADC} 
C_{\rm loc} \left( \adc; \mathfrak{e} \right) = C_{\rm loc, sep} \left( \adc; \mathfrak{e} \right)= \chi(\adc; \mathfrak{e}) = 
  \ergo^{(1)}(\adc;\mathfrak{e})\;,
\end{equation} 
which we plot in Fig.~\ref{fig: plot2} as functions of $\gamma$ for fixed $\mathfrak{e}$.
Notably for the ADC $\adc$ the quantity $\chi(\adc; \mathfrak{e})$ also provides the value of the separable-input ergotropy capacitance, i.e.
\begin{eqnarray} \label{csepadc} 
 C_{\rm sep} \left( \adc; \mathfrak{e} \right)= \chi(\adc; \mathfrak{e})\;.
\end{eqnarray} 
The proof of this important identity is postponed to Sec.~\ref{sec:GADC}, where it will be derived for the larger class of GADC's which includes the $\adc$'s as special instances.

We conclude by noticing that in view of the ordering~(\ref{ordering}),  $\chi(\adc, \mathfrak{e})$ represents a lower bound for 
 the last of the work capacitances introduced in Sec.~\ref{sec:QWC}, i.e. $C_{\cal E} \left( \adc, \mathfrak{e} \right)$.
 An upper bound for such term can be computed exploiting the fact that the mean output energy of a state is always larger than or equal to its ergotropy. Thus, from~(\ref{eq:adcenergylim}), we can write 
\begin{eqnarray} \label{ergonn}
 \ergo^{(n)}(\Phi_\gamma;E)&\leq &  (1-\gamma) E \;,  
 \end{eqnarray} 
 that leads to the inequality
 \begin{eqnarray} \label{csepadcbounds} 
(1-\gamma)  \mathfrak{e} \geq  C_{\cal E} \left( \adc; \mathfrak{e} \right)\geq  \chi(\adc; \mathfrak{e})\;. 
\end{eqnarray} 
Notice that Eq.~(\ref{ergonn}) 
determines the value of $C_{\cal E} \left( \adc; \mathfrak{e} \right)$ at least
 for small enough $\mathfrak{e}$: indeed in this regime 
  $\chi(\adc; \mathfrak{e})$ reduces to
$(1-\gamma)  \mathfrak{e}$, so we can write
  \begin{eqnarray} \label{csepadcbounds1} 
  C_{\cal E} \left( \adc; \mathfrak{e} \right) = (1-\gamma)  \mathfrak{e} + O(\mathfrak{e}^2) \;. 
\end{eqnarray}

\begin{figure*}[t!]
		\centering
		\includegraphics[width=\linewidth]{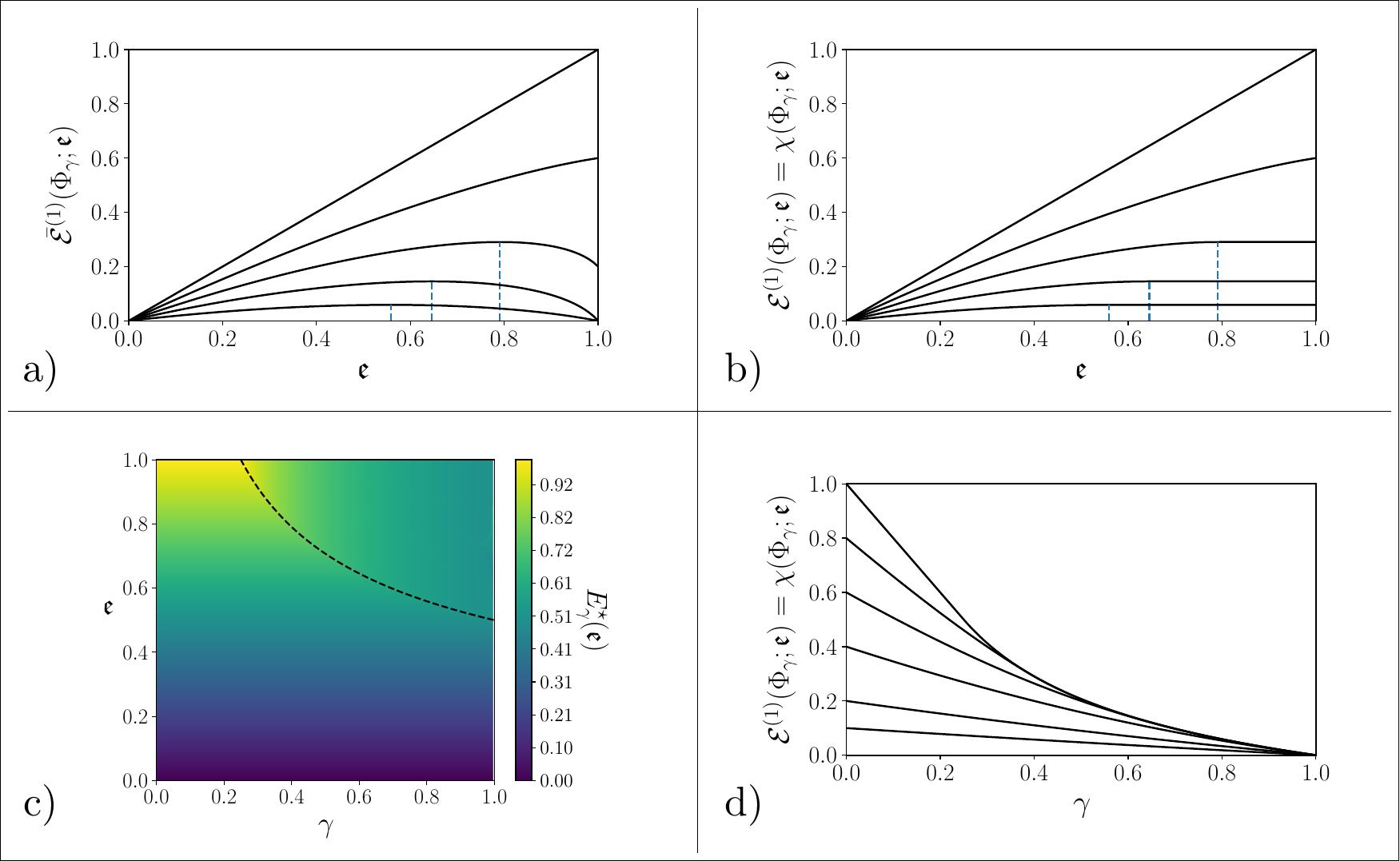}
		\caption{{\bf Panels  a)} and {\bf b)} show respectively the maximal single-shot, fixed energy, output ergotropy  $\fixedergo^{(1)}(\adc;\mathfrak{e})$ of
		Eq.~(\ref{fixedergo_adc}), and  $\ergo^{(1)}(\adc;\mathfrak{e})$ of Eq.~(\ref{eq: ADC erg_11_new}) for the ADC as a function of $\mathfrak{e}$ for different values of the noise parameter $\gamma$ (from top to bottom the values of $\gamma$ have been chosen to be equal to $0, 0.2, 0.4, 0.6, 0.8$).
		 Notice that both  $\fixedergo^{(1)}(\adc;\mathfrak{e})$ and 
		  $\ergo^{(1)}(\adc;\mathfrak{e})$ are concave w.r.t. to $\mathfrak{e}$, and that for $\mathfrak{e}\rightarrow 0$ their first derivatives are equal to $1-\gamma$ (see Eq.~(\ref{torna})). 
		  Recall also that $\ergo^{(1)}(\adc;\mathfrak{e})$ coincides with $\chi(\adc,\mathfrak{e})$, which according to Eq.~(\ref{defchiADC}) and (\ref{csepadc}) provides the value of the local and separable ergotropic capacitances of the channel.  
		  Dotted lines in plots indicate the values of $\mathfrak{e}$ where $\fixedergo^{(1)}(\adc;\mathfrak{e})$ reaches its maximum. 
		{\bf Panel  c)} Average input energy ${E}^{\star}_{\gamma}({\mathfrak{e}})$ of Eq.~(\ref{MAX4}) of the optimal state~(\ref{stato_energiaE})
	that allows $\fixedergo^{(1)}(\adc;\mathfrak{e})$ to be the maximum $\ergo^{(1)}(\adc;\mathfrak{e})$, the dotted line being the curve $\mathfrak{e} = 1/2\sqrt{\gamma}$.
	{\bf Panel d)}  $\ergo^{(1)}(\adc;\mathfrak{e})$ of Eq.~(\ref{eq: ADC erg_11_new}) as a function of the noise parameter $\gamma$ for different values of $\mathfrak{e}$ (from top to bottom $1,0.8,0.6,0.4,0.2,0.1$).}
		\label{fig: plot2}
\end{figure*}

Estimating $C_{\cal E} \left( \adc; \mathfrak{e} \right)$ for arbitrary values of $\mathfrak{e}$ remains however a challenging open problem. The superadditivity, due to entanglement and nonlocal energy extraction operations, that this quantity can in principle exhibit makes the task hard to tackle. Crucially, proving superadditivity would also certify a ``quantum advantage'' of the nonlocal strategies over the classical setting in presence of ADC noise.

 \subsubsection{Estimating MAWERs for ADCs}

 We can exactly compute the value of the MAWER functional $\multiergo(\adc)$ for ADC's. 
 To begin with, observe that from Eq. (\ref{ergonn}) it is easy verified that such quantity cannot be larger than $1-\gamma$, i.e. $\multiergo(\adc) \leq 1-\gamma$.
We will now show a family of states which achieves this upper bound, proving therefore that 
\begin{eqnarray}
\label{mawer_adc}
\multiergo(\adc) = 1-\gamma \; .
\end{eqnarray}
To this aim let's consider, for any fixed energy $E$ and arbitrary $n \geq E$, the vectors 
$\ket{\psi_{E/n}^{(0)}}^{\otimes n}$
 (see~(\ref{stato_energiaE})).
 Observe then that by invoking~(\ref{qubitsimpo}) we can write 
   \begin{eqnarray}\label{impo1} 
\tfrac{\ergo(\adc^{\otimes n}(\ket{\psi_{E/n}^{(0)}}\!\bra{\psi_{E/n}^{(0)}}^{\otimes n});\ham^{(n)})}{E} &=&\tfrac{n  
\ergo(\adc(\ket{\psi_{E/n}^{(0)}}\!\bra{\psi_{E/n}^{(0)}});\hat{h})}{E} \nonumber \\
&=& (1-\gamma) - O\left(\tfrac{E}{n}\right) \; ,
 \end{eqnarray}
 which, taking the sup over $n$, leads to Eq.~(\ref{mawer_adc}). 
We observe next that $1-\gamma$ provides also the value of the constrained versions of the MAWER, i.e. 
\begin{eqnarray}
\label{mawer_adcloc}
{\cal J}_{\rm sep} (\adc)={\cal J}_{\rm loc} (\adc)={\cal J}_{\rm loc, sep} (\adc) = 1-\gamma \; .
\end{eqnarray}
Indeed, while because of~(\ref{mawer_adc}),
$1-\gamma$ is trivially an upper bound for ${\cal J}_{\rm sep} (\adc)$, ${\cal J}_{\rm loc} (\adc)$, ${\cal J}_{\rm loc, sep} (\adc)$, such value can be attained by the smaller of them (i.e. ${\cal J}_{\rm loc, sep} (\adc)$), again using the separable state $\ket{\psi_{E/n}^{(0)}}^{\otimes n}$ as trial input.  
This follows by simply noticing that the identity in Eq.~(\ref{impo1}) would trivially apply also if we replaced 
$\ergo(\adc^{\otimes n}(\ket{\psi_{E/n}^{(0)}}\!\!\bra{\psi_{E/n}^{(0)}};\ham^{(n)})$ in the l.h.s. side with 
$\ergo_{\rm loc}(\adc^{\otimes n}(\ket{\psi_{E/n}^{(0)}}\!\!\bra{\psi_{E/n}^{(0)}};\ham^{(n)})$.
It goes without mentioning that, since the function~$\chi (\adc,  \mathfrak{e})$ verifies the conditions
$\chi (\adc;  \mathfrak{e}=0) =0$ and 
 \begin{eqnarray}\label{torna} 
 \left.\frac{\partial{\chi (\adc;  \mathfrak{e})}}{\partial  \mathfrak{e}}\right|_{\mathfrak{e} =0}&=&
  \left.\frac{\partial{\ergo^{(1)}(\adc;\mathfrak{e})}}{\partial  \mathfrak{e}}\right|_{\mathfrak{e} =0} =1-\gamma\;,
 \end{eqnarray}  
 for the local, separable-input local and separable MAWER, the identity (\ref{mawer_adcloc}) 
exactly matches with the prescription of Theorem~\ref{theo1}.
We also notice that~(\ref{csepadcbounds1}) can be seen as a consequence of Remark~\ref{remark2} and~(\ref{mawer_adc}).

 We conclude by observing that the vectors $\ket{\psi_{E/n}^{(0)}}^{\otimes n}$ are not the only ones that allow us to assign the value~$1-\gamma$ to $\multiergo(\adc)$. 
 Interestingly we can also employ more ``classical''  input states, that at variance with $\ket{\psi_{E/n}^{(0)}}^{\otimes n}$
 exhibit no quantum coherence among the energy eigenstates of the system Hamiltonian. Consider for instance what happens if we take tensor product states of the form 
 \begin{eqnarray} \label{classical} 
  |\phi_E^{(n)}\rangle:=  \ket{1}^{\otimes \lfloor 
 E\rfloor}\otimes\ket{0}^{n-\lfloor E \rfloor}\;,
 \end{eqnarray} 
 in which the first $\lfloor E \rfloor$ q-cells are initialized into the maximum energy eigenstate of $\hat{h}$ while the remaining ones are prepared in the ground state. 
In this case it is still possible to realize a work/energy ratio as high as the MAWER. Indeed  invoking Eq.~(\ref{Tzero}) we can write
 \begin{eqnarray}
&&\lim_{n\to\infty} \en\left(\ADCn (|\phi_E^{(n)}\rangle \!\langle \phi_E^{(n)}|); \ham^{(n)} \right) \nonumber \\
&&\qquad = \lfloor E \rfloor(1-\gamma) \nonumber \\
&&\qquad \geq \lim_{n \rightarrow \infty} \ergo\left(\ADCn (|\phi_E^{(n)}\rangle \!\langle \phi_E^{(n)}|); \ham^{(n)} \right)\nonumber \\ 
&&\qquad = \lim_{n \rightarrow \infty} \ergo\left(\left(\adc(\ket{1}\!\!\bra{1})\right)^{\otimes \lfloor E\rfloor} \otimes\ket{0}\!\!\bra{0}^{\otimes(n-\lfloor E \rfloor)}; \ham^{(n)} \right) \nonumber \\
&& \qquad \geq \en\left(\adc(\ket{1}\!\!\bra{1})^{\otimes \lfloor E \rfloor};\ham^{(\lfloor E \rfloor)}\right) - 1 \nonumber \\ 
&& \qquad = \lfloor E\rfloor \en\left(\adc(\ket{1}\!\!\bra{1}); \hat{h}\right) - 1 \nonumber \\
&& \qquad = \lfloor E\rfloor (1-\gamma) - 1 \; ,
 \end{eqnarray}
where the third passage is justified because the rank of the state $\adc(\ket{1}\!\!\bra{1})^{\otimes \lfloor E\rfloor} \otimes\ket{0}\!\!\bra{0}^{\otimes(n-\lfloor E \rfloor)}$ is at most $2^{\lfloor E \rfloor}$, so for $n\geq 2^{\lfloor E \rfloor} + \lfloor E \rfloor$ we can find a suitable work extraction unitary rearranging all the output state eigenvalues into the first excited state of the Hamiltonian $\ham^{(n)}$. Therefore, remembering that the supremum over $n$ can always be replaced by a limit (see Eq.~(\ref{eq:suplim}) of Appendix~\ref{app:physmaw}), we can write 
 \begin{eqnarray}
 \label{adc_classicalstrategy}
&\limsup\limits_{E\to\infty} \sup\limits_{n \geq \lceil E\rceil} \frac{\ergo\left(\ADCn (\ket{1}\!\bra{1}^{\otimes \lfloor E\rfloor}\otimes\ket{0}\!\bra{0}^{n-\lfloor E \rfloor}); \ham^{(n)} \right)}{E}
 \nonumber \\
 &= 1-\gamma = \multiergo(\adc) \; .
 \end{eqnarray}
 However, while the family $\ket{\psi_{E/n}^{(0)}}^{\otimes n}$ allows to reach the limit~(\ref{mawer_adc}) with local unitaries (see discussion below~(\ref{mawer_adcloc})), if we use the ``classical'' input states $|\phi_E^{(n)}\rangle$ we necessarily need the power of nonlocal operations to extract the energy at the output of the channel. Indeed using only local operations in this case we  have 
 \begin{eqnarray}
 \ergo_{\rm loc} \left(\ADCn (|\phi_E^{(n)}\rangle \!\langle \phi_E^{(n)}|); \ham^{(n)} \right) &=& 
 \lfloor E\rfloor \ergo \left(\ADCn (|1\rangle\!\langle 1|); \hat{h} \right) \nonumber \\
 &=& \lfloor E\rfloor \max\left(0, 1-2\gamma\right) \;, \nonumber  
 \end{eqnarray}
 which leads to  
 \begin{eqnarray}
&\limsup\limits_{E\to\infty} \sup\limits_{n \geq \lceil E\rceil} \frac{\ergo_{\rm loc}\left(\ADCn (\ket{1}\!\bra{1}^{\otimes \lfloor 
 E\rfloor}\otimes\ket{0}\!\bra{0}^{n-\lfloor E \rfloor}); \ham^{(n)} \right)}{E}
 \nonumber \\
 &= \max\left(0, 1-2\gamma\right) < \multiergo(\adc) \; .
 \label{ergoloc_gadc}
 \end{eqnarray}
 It is finally worth noting that any state $\dstate^{(n)}_d$ which is diagonal in the same basis as $\ham^{(n)}$ can be written as a linear convex combination of terms in the form $\ket{1}\!\!\bra{1}^{\otimes k}\otimes\ket{0}\!\!\bra{0}^{n-k}$. Therefore, Eq.~(\ref{ergoloc_gadc}) implies that no diagonal input state can attain the MAWER for this family of channels.
 
The take home message of all the here presented discussion is that the family of states $\ket{\psi_{E/n}^{(0)}}^{\otimes n}$ is capable of reaching the maximal MAWER efficiency under the effect of the ADC noise using only local operations: this thanks to the explicit
 presence of quantum coherence in the superposition~(\ref{stato_energiaE}).

\subsection{Generalized amplitude damping channel}\label{sec:GADC} 

\begin{figure*}[t!]
    \centering
    \includegraphics[width=\linewidth]{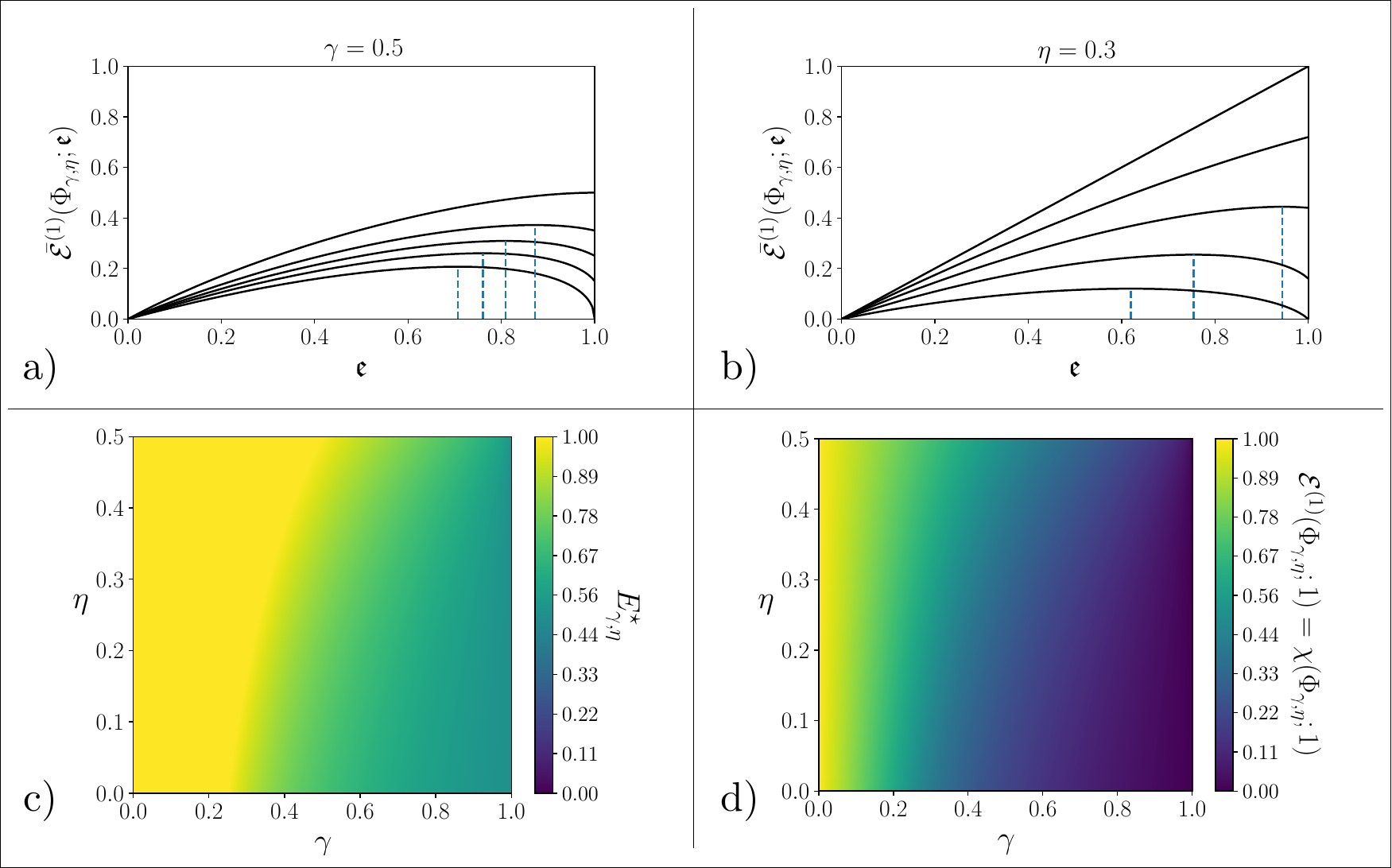}
    \caption{{\bf Panel a)} Maximal single-shot output ergotropy at fixed input energy $\bar{\ergo}^{(1)}(\gadc,\mathfrak{e})$ (see Eq. (\ref{eq: GADC erg_E})) as a function of $\mathfrak{e}$ for the GADC for different values of the parameter $\eta$ (from bottom to top [0, 0.15, 0.25, 0.35, 0.5]), by keeping $\gamma=0.5$. {\bf Panel b)} $\bar{\ergo}^{(1)}(\gadc,\mathfrak{e})$ (see Eq. (\ref{eq: GADC erg_E})) for different values of $\gamma$ (from top to bottom [0, 0.2, 0.4, 0.6, 0.8]), with $\eta=0.3$. Blue dotted lines indicate the values of $\mathfrak{e}$ maximizing $\bar{\ergo}^{(1)}(\gadc,\mathfrak{e})$. Notice that for any value of the parameters the output ergotropy is a concave function of $\mathfrak{e}$. {\bf Panel c)} Average input energy $E^{\star}_{\gamma,\eta}$ Eq. (\ref{MAX4new11}) of the optimal state (\ref{stato_energiaE}) maximizing the single-shot output ergotropy. {\bf Panel d)} Values of the single-shot output ergotropy $\ergo^{(1)}(\gadc;1)$ (see Eq. (\ref{eq: GADC erg_Emax})) for any values of the channel parameters $\gamma$ and $\eta$. We recall that due to its concavity $\ergo^{(1)}(\gadc;1) = \chi(\gadc;1)$, in the case of the GADC also $\ergo^{(1)}(\gadc;1) = C_{\rm sep}(\gadc;1)$.} 
    \label{fig:mergeGADC}
\end{figure*}

The action of a GADC $\gadc$ on a density matrix $\dstate$ is defined in terms of two noise parameters $\gamma$ and $\eta$ via the mapping 
     \cite{qubitGADC}:
    \begin{equation}\label{eq: GADC matrix}
	\gadc(\dstate) = \begin{pmatrix}
	1 - \gamma\eta - (1-\gamma)\rho_{11}  &  \sqrt{1-\gamma} \rho_{01} \\
	\sqrt{1-\gamma} \rho_{10}  & \gamma\eta + (1-\gamma)\rho_{11}
	\end{pmatrix}.
	\end{equation}
The channel $\gadc$ is in principle well defined for every $\eta \in [0,1]$. However, it represents a thermalization at finite positive temperature only when $\eta \in \left[0,{1}/{2}\right]$, we will henceforth assume this as the proper parameter space.
 Notice in particular that for $\eta=0$, the GADC reduces to a conventional ADC, i.e. $\Phi_{\gamma,0} = \Phi_\gamma$. For $\eta=1/2$ instead it describes the thermal contact with an infinite temperature. Additionally, for arbitrary $\eta \in \left[0,{1}/{2}\right]$ $\gadc$ sends the ground state into a thermal state~(\ref{GIBBS}) of inverse temperature
 \begin{eqnarray} 
 \beta(\gadc) = \beta_{\gamma,\eta} := - \ln\left(\frac{\gamma\eta}{1-\gamma \eta}\right)\;. 
 \end{eqnarray} 
  Similarly to the case of ADCs one can easily verify that GADCs commute with the Hamiltonian evolution
 (i.e. $\gadc(e^{-i\hat{h} t}\dstate e^{i\hat{h} t}) = e^{-i\hat{h} t}\gadc(\dstate)e^{i\hat{h} t}$) and maintain linear dependence between the input and output energies, i.e. 
 \begin{equation} \label{eq:adcenergylimGADC}
\en(\gadc(\dstate)) = (1-\gamma)\en(\dstate) + \gamma \eta \; .
\end{equation}

As for the ADC, we can compute the single-shot fixed-energy maximum output ergotropy
focusing on the input states~(\ref{stato_energiaE}).

In this case 
the eigenvalues of $\gadc(\ket{\psi_{\mathfrak{e}}^{(\phi)}}\!\!\bra{\psi_{\mathfrak{e}}^{(\phi)}})$ are given by
	\begin{equation}
   \label{eigenvalues_gadc}	 \lambda^{({\pm})}_{\gamma,\eta}({\mathfrak{e}})
   	 := \frac{1}{2} \left(1\pm \sqrt{(1-2\eta\gamma)^2+4\gamma(1-\gamma)(2\eta-{\mathfrak{e}}){\mathfrak{e}}}\right),
	\end{equation}
which replaced into~(\ref{ergoqubit}) lead to
	\begin{eqnarray}\label{eq: GADC erg_E}
	&&\fixedergo^{(1)}(\gadc;{\mathfrak{e}}) =  (1-\gamma){\mathfrak{e}} + \eta\gamma \\ &&\qquad - \frac{1}{2} \left(
	1- \sqrt{(1-2\eta\gamma)^2+4\gamma(1-\gamma)(2\eta-{\mathfrak{e}}){\mathfrak{e}}}\right)\; . \nonumber
	\end{eqnarray}
The behavior of $\fixedergo^{(1)}(\gadc;{\mathfrak{e}})$ is depicted in Fig. \ref{fig:mergeGADC}. Simple algebra reveals that, for fixed $\eta$ and $\gamma$, such a function is still concave w.r.t. to $\mathfrak{e}$ and that, on the interval $[0,1]$, it achieves its maximum for 
	\begin{eqnarray}
	E_{\gamma,\eta} := \min\left\{ 1, \eta + \tfrac{\sqrt{1-4 \gamma \eta (1-\eta)}}{2 \sqrt{ \gamma}}\right\}\;. 
	 \end{eqnarray}
Accordingly, introducing 
	\begin{equation}
{E}^{\star}_{\gamma,\eta}({\mathfrak{e}}):= \min\{ {\mathfrak{e}}, E_{\gamma,\eta}\} =
\label{MAX4new11} 
\min\left\{ {\mathfrak{e}}, \eta + \tfrac{\sqrt{1-4 \gamma \eta (1-\eta)}}{2 \sqrt{ \gamma}}\right\}\;,
\end{equation} 
we can write 
	 	\begin{align}\label{eq: GADC erg_Emax}
	\ergo^{(1)}(\gadc;{\mathfrak{e}}) = \fixedergo^{(1)}(\gadc; {E}^{\star}_{\gamma,\eta}({\mathfrak{e}}))\; ,
	\end{align}
	which inherits from (\ref{eq: GADC erg_E}) the property of being a concave function of ${\mathfrak{e}}$.
Replaced in Eqs.~(\ref{exact}) and~(\ref{defchi}) this finally gives
\begin{eqnarray} \label{defchiGADC} 
C_{\rm loc} \left( \gadc; \mathfrak{e} \right) &=& C_{\rm loc, sep} \left( \gadc; \mathfrak{e} \right)\nonumber \\
&=& \chi(\gadc; \mathfrak{e}) = 
  \ergo^{(1)}(\gadc;\mathfrak{e})\;,
\end{eqnarray} 
which we plot in Fig.~\ref{fig:mergeGADC} for $\mathfrak{e}$ alongside with the values of the optimal input energy $E^{\star}_{\gamma,\eta}$. 
As anticipated in the previous section,
for the GADCs $\gadc$ the quantity $\chi(\gadc, \mathfrak{e})$ also provides the value of the separable-input ergotropy capacitance,
i.e.
\begin{eqnarray} \label{csepgadc} 
 C_{\rm sep} \left( \gadc; \mathfrak{e} \right)= \chi(\gadc; \mathfrak{e})\;. 
\end{eqnarray} 
An explicit proof of this result follows from two facts. The first is simply that, since for qubit systems the total ergotropy and the ergotropy always coincide (see Eq.~(\ref{qubitimpo1}) of Appendix~\ref{sec.rev}), we can claim
$\chi_{\rm tot} (\gadc, \mathfrak{e})= \chi(\gadc, \mathfrak{e})$. 
The second is that GADCs fulfill the condition~(\ref{step1}) which, according to Ref.~\cite{tirone2023work}, enables one to show that the r.h.s. side of Eq.~(\ref{defchitot}) coincides with  $C_{\rm sep} \left(\gadc, \mathfrak{e} \right)$.
To verify this fact, notice that by definition the single-site states entering (\ref{defPSIfac}) can be identified as states as in~(\ref{stato_energiaE}), i.e. 
\begin{eqnarray} 
|\psi_j\rangle = \ket{\psi_{\mathfrak{e}_j}^{(\phi_j)}}\;,
\end{eqnarray} 
with $\phi_j$ arbitrary, and energy terms ${\mathfrak{e}_j}$ such that 
\begin{eqnarray} 
\overline{\mathfrak{e}}:=  \sum_{j=1}^n {\mathfrak{e}_j}/n  \leq  {\mathfrak{e}}\;. 
\end{eqnarray} 
Observe next that, by explicit calculation, the von Neumann entropy of these states evaluated at the output of the GADC turns out to be independent of $\phi_j$ and to be a convex function of ${\mathfrak{e}_j}$,
\begin{eqnarray} \label{entropyeta} 
&&S(\gadc(|\psi_{\mathfrak{e}_j}^{(\phi_j)}\rangle \!\langle\psi_{\mathfrak{e}_j}^{(\phi_j)}|)) \\
&&\qquad \qquad = H_{\gamma,\eta}({\mathfrak{e}_j}) \nonumber 
:= - \sum_{k=\pm} \lambda_{\gamma,\eta}^{(k)}(\mathfrak{e}_j) \log_2  \lambda_{\gamma,\eta}^{(k)}(\mathfrak{e}_j) \;,\end{eqnarray} 
(see Fig.~\ref{fig:entropy}).

Thus we can write 
        \begin{eqnarray}
        S\left(\gadc^{\otimes n}(|\Psi_{\rm fact}^{(n)}\rangle\!\langle \Psi_{\rm fact}^{(n)}|) \right)
        &=&\sum_{j=1}^n  S\left(\gadc(\ket{\psi_{\mathfrak{e}_j}^{(\phi_j)}}\!\!\bra{\psi_{\mathfrak{e}_j}^{(\phi_j)}} )\right)
            \nonumber \\
            &=&   \sum_{j=1}^n H_{\gamma,\lambda}({\mathfrak{e}_j})  
            \nonumber \\
            &\geq&  n  H_{\gamma,\lambda}(\sum_{j=1}^n {\mathfrak{e}_j}/n) \nonumber \\
            &=&n  S\left(\gadc(\ket{\psi_{\overline{\mathfrak{e}}}^{(0)}}\!\!\bra{\psi_{\overline{\mathfrak{e}}}^{(0)}} )\right)
            \nonumber \\
            &=& S\left((\gadc(\ket{\psi_{\overline{\mathfrak{e}}}^{(0)}}\!\!\bra{\psi_{\overline{\mathfrak{e}}}^{(0)}}))^{\otimes n} \right) \;, \nonumber 
        \end{eqnarray}
which explicitly proves that the state $\ket{\psi_{\overline{\mathfrak{e}}}^{(0)}}^{\otimes n}$ minimizes the output entropy of the channel $\gadc^{\otimes n}$ over the set of pure states $|\Psi_{\rm fact}^{(n)}\rangle$ of fixed input energy $n\overline{\mathfrak{e}}$.
        As the entropy and the total ergotropy are related by a bijective relation (see e.g.~\cite{PhysRevA.105.012414}) the above result also implies that the equipartite state $\ket{\psi_{\overline{\mathfrak{e}}}^{(0)}}^{\otimes n}$ is the one that maximizes the output total ergotropy $\ergo_{\rm tot}(\gadc(|\Psi_{\rm fact}^{(n)}\rangle \!\langle \Psi_{\rm fact}^{(n)}|))$ for all such states. Setting $\hat{\sigma}_{\mathfrak{e}} =  \ket{\psi_{\overline{\mathfrak{e}}}^{(0)}}\!\! \bra{\psi_{\overline{\mathfrak{e}}}^{(0)}}$ we can hence write
        \begin{eqnarray} 
\ergo_{\rm tot}((\gadc(\hat{\sigma}_{\mathfrak{e}}))^{\otimes n};\ham^{(n)}) &\geq& \ergo_{\rm tot}(\gadc^{\otimes n}(|\Psi_{\rm fact}^{(n)}\rangle\!\langle \Psi_{\rm fact}^{(n)}|);\ham^{(n)}) \nonumber \\ &\geq& 
\ergo(\gadc^{\otimes n}(|\Psi_{\rm fact}^{(n)}\rangle\!\langle \Psi_{\rm fact}^{(n)}|);\ham^{(n)}) \nonumber \end{eqnarray} 
which leads to Eq.~(\ref{step1}) by the following identities 
\begin{eqnarray} \nonumber 
&& \ergo_{\rm tot}((\gadc(\hat{\sigma}_{\mathfrak{e}}))^{\otimes n};\ham^{(n)})  = n \ergo_{\rm tot}(\gadc(\hat{\sigma}_{\mathfrak{e}});\hat{h})  \\
&&\qquad \qquad = n \ergo(\gadc(\hat{\sigma}_{\mathfrak{e}});\hat{h}) = \ergo((\gadc(\hat{\sigma}_{\mathfrak{e}}))^{\otimes n};\ham^{(n)})\;,\nonumber
\end{eqnarray} 
        the first being a general property of the total ergotropy, the last two applying instead for the special case of qubits thanks to Eq.~(\ref{qubitsimpo}).
        
        We conclude by remarking that as in the case of the ADC's we cannot provide 
        the value of the ergotropic capacitance of GADCs and observing that in this case Eq.~(\ref{csepadcbounds}) becomes        \begin{eqnarray} \label{csepadcboundsnew} 
(1-\gamma)  \mathfrak{e} + \gamma \eta  \geq  C_{\cal E} \left( \gadc; \mathfrak{e} \right)\geq  \chi(\gadc; \mathfrak{e})\;.
\end{eqnarray} 

\begin{figure}[h!]
\centering
\includegraphics[width=0.95\linewidth]{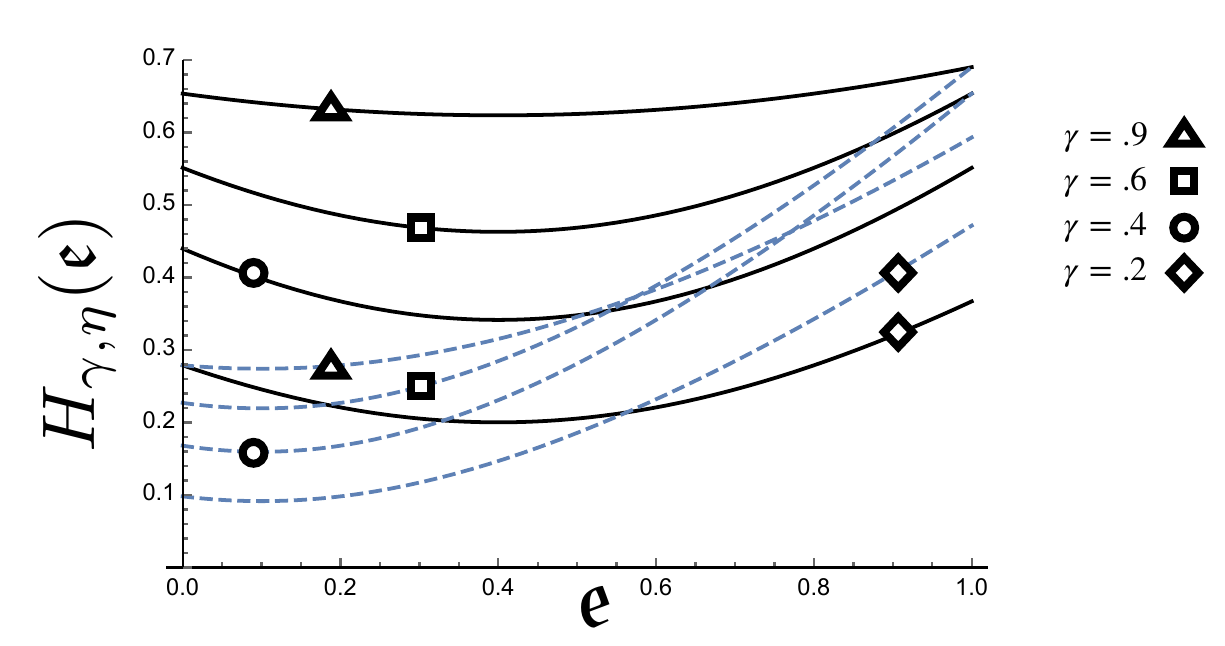}
\caption{Plot of the output entropy $H_{\gamma,\eta}({\mathfrak{e}})$~(\ref{entropyeta}) of the state~$|\psi_{\mathfrak{e}}^{(\phi)}\rangle$ 
emerging from the GADC $\gadc$ as a function of  ${\mathfrak{e}}$ for $\eta=0.1$ (dashed curves),
$\eta=
0.4$ (black curves),  and various values of $\gamma$.
}
\label{fig:entropy}
\end{figure}

\subsubsection{Estimating MAWERs for GADCs}
Theorem~\ref{theo1} and Eqs.~(\ref{defchiADC}) and (\ref{csepadc}) allow us to explicitly evaluate the local and separable-input MAWERs in terms of the first derivative of $\ergo^{(1)}(\gadc;\mathfrak{e})$ at $\mathfrak{e}=0$: 
\begin{equation}{\cal J}_{\rm sep}(\gadc) ={\cal J}_{\rm loc}(\gadc) =
{\cal J}_{\rm loc, sep}(\gadc) = \frac{1-\gamma}{1-2\eta\gamma}  \; . \label{eq:cohstrat_gadcmaw}
\end{equation}
Unfortunately, differently from the case of the amplitude damping channel $\adc$,  
we do not have upper bounds allowing us to establish that the r.h.s. of  (\ref{eq:cohstrat_gadcmaw}) coincides with the unrestricted ergotropic MAWER ${\cal J}_{\cal E}(\gadc)$.
 An interesting feature of the model emerges by noticing that, since for $\eta\in (0,1/2]$ 
the term $\frac{1-\gamma}{1-2\eta\gamma}$ is always strictly larger than $1-\gamma$, we can claim that 
\begin{equation}
{\cal J}_{\rm loc, sep}(\gadc) > {\cal J}_{\cal E}(\adc) \;, \qquad \forall \eta >0\;,
   \label{eq:cohstrat_gadcmaweewe}
\end{equation}
indicating that the presence of thermal noise in the model can help with the work extraction process. 
The gap in Eq.~(\ref{eq:cohstrat_gadcmaweewe}) becomes particularly significative in the infinite temperature limit 
($\eta = 1/2$) where, according to  (\ref{eq:cohstrat_gadcmaw}), we get 
  \begin{eqnarray} {\cal J}_{\rm loc, sep}(\Phi_{\gamma, \eta=1/2}) = 1\;, \qquad \forall \gamma\in [0,1)\;,\end{eqnarray} 
  indicating that in this case one can recover the same amount of energy as the one initially stored into the QB.
  
By explicit evaluation it turns out that the value~(\ref{eq:cohstrat_gadcmaw}) can be achieved by taking as input the family of factorized states $\ket{\psi_{E/n}^{(0)}}^{\otimes n}$ as in the notation of Eq.~(\ref{stato_energiaE}). Following the same analysis of Eq.~(\ref{impo1}) we can in fact write
  \begin{eqnarray}\label{impo1newnew} 
&&\tfrac{\ergo(\gadc^{\otimes n}(\ket{\psi_{E/n}^{(0)}}\!\bra{\psi_{E/n}^{(0)}}^{\otimes n});\ham^{(n)})}{E} =\tfrac{n  
\ergo(\gadc(\ket{\psi_{E/n}^{(0)}}\!\bra{\psi_{E/n}^{(0)}});\hat{h})}{E} \nonumber \\
&& = (1 - \gamma) +\tfrac{2\gamma(1-\gamma)\eta}{1-2\eta\gamma} + O\left(\tfrac{E}{n}\right)=
\tfrac{1-\gamma}{1-2\eta\gamma} + O\left(\tfrac{E}{n}\right), 
 \end{eqnarray}
 which in the $n \to \infty$ limit leads to~(\ref{eq:cohstrat_gadcmaw}). 
On the contrary no analogue of Eq.~(\ref{adc_classicalstrategy}) is true for the channel $\gadc$ when $\eta > 0$: at finite temperature the ``classical'' input states $|\phi_E^{(n)}\rangle$ of~(\ref{classical}) yield a work/energy ratio strictly smaller than the one granted by the coherent input state $\ket{\psi_{E/n}^{(0)}}^{\otimes n}$.
\begin{figure}[t!]
	\centering
	\includegraphics[width=\linewidth]{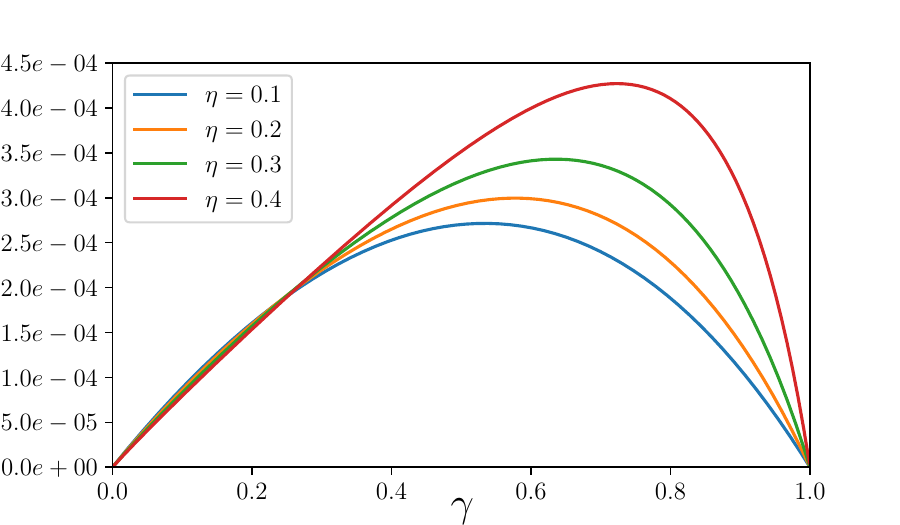}
	\caption{Difference between
	${\cal J}_{\rm loc, sep}(\gadc)$ of (\ref{eq:cohstrat_gadcmaw}) and 
	the asymptotic work/energy ration achievable with the ``classical''  states $|\phi_E^{(n)}\rangle$ computed in
(\ref{impo1newnewasdf}) 
	 as a function of the damping parameter $\gamma$ for various values of $\eta$.}
    \label{fig:gadc_mawer_cfr}
	\end{figure}
To verify this we observe that the (possibly nonlocal) asymptotic work/energy ratio achievable with the states $|\phi_E^{(n)}\rangle$
 can be upper bounded as
  \begin{eqnarray}&&  \label{impo1newnewasdf} 
  \limsup_{E\rightarrow \infty} \lim_{n\rightarrow \infty}
  \tfrac{\ergo(\gadc^{\otimes n}(\ket{\phi^{(n)}_{E}}\!\bra{\phi^{(n)}_{E}});\ham^{(n)})}{E}
 \\
&&\qquad = \limsup_{E\rightarrow \infty} \lim_{n\rightarrow \infty} 
\tfrac{\ergo(\left(\gadc(\ket{1}\!\bra{1})\right)^{\otimes\lfloor E \rfloor} \otimes \hat{\tau}_{\beta_{\gamma,\eta}}^{\otimes (n-\lfloor E \rfloor)};\ham^{(n)})}{E}\nonumber  \\ \nonumber 
&&\qquad \leq\limsup_{E\rightarrow \infty}
 \tfrac{{\cal W}_{\beta_{\gamma,\eta}}(\left(\gadc(\ket{1}\!\bra{1})\right)^{\otimes\lfloor E \rfloor} ; \ham^{(\lfloor E \rfloor)})}{E}\\ 
&& \qquad =\left(\limsup_{E\rightarrow \infty} \nonumber 
\frac{\lfloor E \rfloor }{E} \right) {\cal W}_{\beta_{\gamma,\eta}}(\gadc(\ket{1}\!\!\bra{1}) ; \hat{h}) 
\\ && \qquad ={\cal W}_{\beta_{\gamma,\eta}}(\gadc(\ket{1}\!\!\bra{1}) ; \hat{h}) \nonumber \\ \nonumber 
&& \qquad = (1-\gamma) + \gamma \eta- \tfrac{H_{\gamma,\eta}(1)-\ln (1+ e^{-\beta_{\gamma,\eta}})}{\beta_{\gamma,\eta}} \;, 
 \end{eqnarray}
 where in the third passage we exploited results in \cite{obejko2021}, here $H_{\gamma,\eta}(1)$ being the output entropy of $|1\rangle$ as computed in (\ref{entropyeta}); here $\mathcal{W}_{\beta_{\gamma,\eta}}$ is the energy extractable with the use of a perfect thermal bath with inverse temperature $\beta_{\gamma,\eta}$ as defined in \cite{quantumworkcapacitances}.
  A comparison of~(\ref{impo1newnewasdf}) with the formula
 (\ref{eq:cohstrat_gadcmaw}), which expresses ${\cal J}_{\rm loc, sep}(\gadc)$, reveals that, apart from the trivial cases where $\gamma=0,1$, the former is always strictly smaller than the latter -- see e.g. Fig.~\ref{fig:gadc_mawer_cfr}.
In Fig.~\ref{fig: GADC_Diff_asympt_Erg1} instead we report the relative difference
\begin{eqnarray}\label{reld} 
\Delta(\gadc) := \frac{{\cal J}_{\rm loc, sep}(\gadc)-r(\gadc)}{r(\gadc)}\;,
\end{eqnarray} 
where 
 \begin{eqnarray}
 r(\gadc):=
 \frac{ \ergo^{(1)}(\gadc; 1)}{{E}^{\star}_{\gamma,\eta}(1)}\;, 
 \end{eqnarray}
  is the ratio between the single-shot maximum output ergotropy~(\ref{eq: GADC erg_Emax}) at full input energy (i.e. the separable-input work capacitance ${C_{\rm sep}( \gadc, 1 )}$) and the corresponding optimal 
 energy input~(\ref{MAX4new11}).
 As evident from the plot, such quantity is strictly positive in a vast volume of the parameter region. This points out that the asymptotic strategy maximizing the local MAWER  performs better than a direct strategy in which 
  one blindly utilizes all the q-cells to store as much input energy as possible. In the previous analysis we have fixed the parameter $\gamma$, which is linked to the discharging time, and we proved that for any fixed $\gamma$ coherence enhances the energetic efficiency. If we instead fix the energy that we want to preserve, the same argument guerantees us that it will be preserved for longer times.

\begin{figure}[t!]
\centering
\includegraphics[width = 1\linewidth]{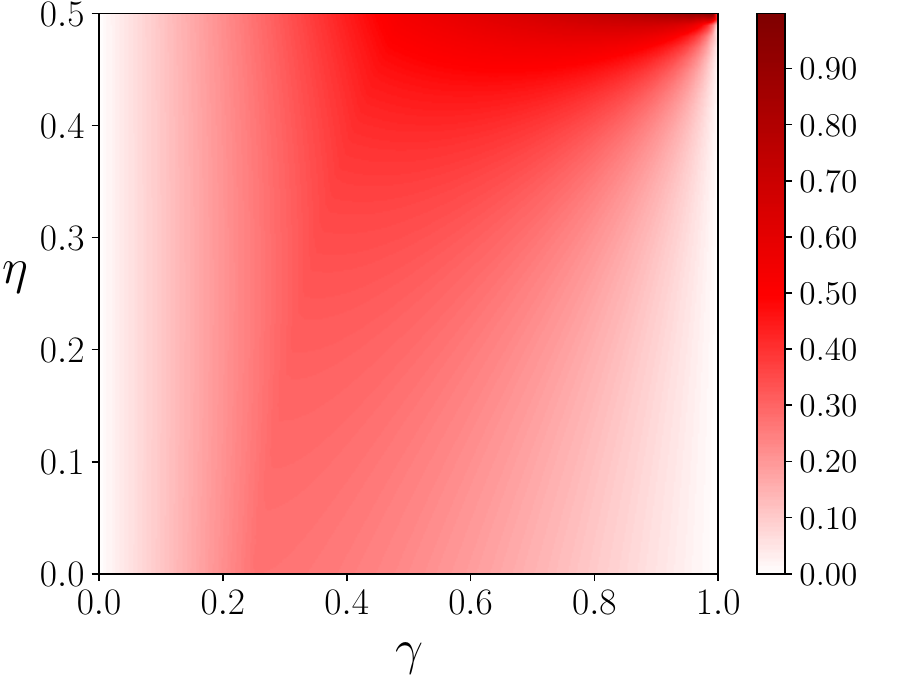}
\caption{Plot of the relative difference $D(\gadc)$ of Eq.~(\ref{reld}) between the ${\cal J}_{\rm loc, sep}(\gadc)$ and the ratio $r(\gadc)$ for different values of the noise parameters $\gamma$, $\eta$. }
\label{fig: GADC_Diff_asympt_Erg1}
\end{figure}

\subsubsection*{Two-uses output ergotropy}
In this section we investigate possible superadditivity effects of the GADC. We do so by an extensive numerical calculation on two-qubits states, finding no violation of additivity in the two-uses case.
We performed a numerical search looking for two-qubit states $\rho^{(2)}$ such that 
\begin{eqnarray}
\ergo(\gadc^{\otimes 2})\left( \rho^{(2)} \right) > 2\ergo^{(1)}\left( \gadc ; \en(\rho^{(2)} ) /2\right) \; .
\label{superadditivity_gadc}
\end{eqnarray}
Any pure two-qubit state $\ket{\psi}\!\!\bra{\psi}$ can be parameterized as
\begin{eqnarray}
\ket{\psi} = \sum_{i=1}^4 c_i \ket{i} e^{i\phi_i} \; ,
\end{eqnarray}
with $c_i \in [0,1]$ and $\phi_i \in [0, 2\pi)$.
Since we know that the spectrum of the output state $\gadc\left( \ket{\psi}\!\!\bra{\psi} \right)$ is independent of the phases $\phi_i$, we can set without loss of generality $\phi_i = 0$.
Using the normalization constraint $\sum_i c^2_i = 1$ and the average energy constraint $\braket{\phi|E|\phi}$, we can restrict the search to a sample space of two real parameters. Among these states, to test possible additivity violations, we selected separable states and entangled input states of the form $\sqrt{1-E/2}
\ket{00}+\sqrt{E/2}\ket{11}$ and compared the associated output ergotropy (at fixed average input energy $\mathfrak{e}$). The small number of real parameters involved in the optimization allows a brute-force span of the parameter space. The numerical search failed though to find, for any value of $\gamma$ and $\eta$ (including the ADC case $\eta=0$), any instance $\rho$ verifying the strict inequality~(\ref{superadditivity_gadc}). Therefore, while non exhaustive, numerical evidence seems to suggest that $\ergo^{(2)}\left( \gadc; E \right) = \ergo^{(1)}\left( \gadc; E \right)$. Considering that the same technique was sufficient to find violations for other qubit channels (e.g. dephasing \cite{quantumworkcapacitances}), we conjecture the optimal output ergotropy of the GADC (and ADC) to be additive; i.e. that $\ergo^{(n)}\left( \gadc; E \right) = n\ergo^{(1)}\left( \gadc; E \right)$ for every $n$ and $E$.

	\subsection{Composition of dephasing channels and ADCs} \label{sec:metaphysics}

Adding extra dephasing to an ADC $\adc$ leads to the identification of a class of maps defined by the transformations~\cite{JIANG_2012} 
    \begin{equation}\label{eq: ADCDEPH matrix}
	D_{\kappa,\gamma} (\dstate) = \begin{pmatrix}
	\rho_{00} + \gamma \rho_{11} & \sqrt{1-\kappa}\sqrt{1-\gamma}\rho_{01} \\
	\sqrt{1-\kappa}\sqrt{1-\gamma}\rho_{10} & (1-\gamma)\rho_{11}
	\end{pmatrix}\;, 
	\end{equation}
	with $\kappa\in[0,1]$ being the dephasing parameter ($\kappa=0$ representing zero extra noise while 
	$\kappa=1$ representing complete extra dephasing). 
	These maps still commute with the Hamiltonian evolution of the system and have the same input-output energy relations of the ADC's, i.e. 
	  \begin{equation} \label{eq:adcenergylimGADCdep}
\en(D_{\kappa,\gamma}(\dstate)) =\en(\adc(\dstate)) = (1-\gamma)\en(\dstate)  \; .
\end{equation}
The optimization of the 
single-shot fixed-energy maximal output ergotropy proceeds along the same lines 
of the other classes of maps we have analyzed so far. In particular, also in this case we can restrict the analysis to input states~(\ref{stato_energiaE}) whose outputs have eigenvalues
	\begin{equation}
	\lambda^{(\pm)}_{\kappa,\gamma}({\mathfrak{e}}) := \frac{1}{2}\left(1 \pm \sqrt{1-4(1-\gamma){\mathfrak{e}}[{\mathfrak{e}}(\gamma - \kappa) + \kappa)]}\right) \; .
	\end{equation}
 Eq.~(\ref{ergoqubit}) hence allows us to express the fixed energy single-shot maximal ergotropy as
    \begin{eqnarray}\label{eq_app ADCDEPH ergotropy}
    \fixedergo^{(1)}(D_{\kappa,\gamma}; {\mathfrak{e}}) &=& (1-\gamma){\mathfrak{e}} - \frac{1}{2} \\ \nonumber   
	&&+ \frac{1}{2}\sqrt{1-4(1-\gamma){\mathfrak{e}}[{\mathfrak{e}}(\gamma - \kappa) + \kappa)]} \; ,
	\end{eqnarray}
 which, at variance with what happens for ADCs and GADCs, is not always concave w.r.t. to ${\mathfrak{e}}$. 
Indeed computing the second order derivative we observe that 
	\begin{eqnarray} \label{eq:domain}
	&& 
\tfrac{\partial^2\fixedergo^{(1)}(D_{\kappa,\gamma}; {\mathfrak{e}})}{\partial^2 E}\\ 
	\nonumber &&	\quad =
	\tfrac{-2(1-\gamma)(1-\kappa) (\gamma-\kappa (1-\gamma))}{\left(1-4(1-\gamma){\mathfrak{e}}[{\mathfrak{e}}(\gamma - \kappa) + \kappa)]\right)^{{3/2}}}\geq 0  \; 
	 \Longleftrightarrow \; \kappa \geq \tfrac{\gamma}{1-
	\gamma} \;, 
	\end{eqnarray} 
which implies 
	\begin{eqnarray} 
	\begin{cases}
	\kappa < \tfrac{\gamma}{1-
	\gamma} \qquad \Longrightarrow \qquad \fixedergo^{(1)}(D_{\kappa,\gamma}; {\mathfrak{e}}) \quad \mbox{concave}  \\
	\kappa = \tfrac{\gamma}{1-
	\gamma} \qquad \Longrightarrow \qquad \fixedergo^{(1)}(D_{\kappa,\gamma}; {\mathfrak{e}}) \quad \mbox{linear} \\
	\kappa >  \tfrac{\gamma}{1-
	\gamma} \qquad \Longrightarrow \qquad \fixedergo^{(1)}(D_{\kappa,\gamma}; {\mathfrak{e}}) \quad \mbox{convex.} 
	\end{cases} 
	\end{eqnarray} 
	Considering that $\fixedergo^{(1)}(D_{\kappa,\gamma}; {\mathfrak{e}})$ is always positive semidefinite (with 
	$\fixedergo^{(1)}(D_{\kappa,\gamma}; 0)=0$), this means that for $\kappa \geq  \tfrac{\gamma}{1-\gamma}$, the function~(\ref{eq_app ADCDEPH ergotropy}) is monotonically increasing, achieving its maximum value
	 \begin{eqnarray}\label{eq_app ADCDEPH ergotropymax}
    \fixedergo^{(1)}(D_{\kappa,\gamma}; 1) &=& \max\left(0, 1-2\gamma\right) \; ,
	\end{eqnarray}
	at the extreme of the allowed domain (i.e. for ${\mathfrak{e}}=1$). Here, $\kappa \geq \frac{\gamma}{1-\gamma}$ is possible for $\gamma \leq 1/2$, as it can be seen in Eq.~\eqref{eq:domain}.
	Accordingly we can conclude that 
	\begin{equation} 
	 \kappa  \in [ \tfrac{\gamma}{1-
	\gamma},1]  \; \Longrightarrow \; \left\{
	\begin{array}{l}
	E_{\kappa,\gamma}^{\star}({\mathfrak{e}}) = E \;, \\ \\ 
	\ergo^{(1)}(D_{\kappa,\gamma}; {\mathfrak{e}})= \fixedergo^{(1)}(D_{\kappa,\gamma}; {\mathfrak{e}})\;, \\ \\
	\chi(D_{\kappa,\gamma},\mathfrak{e}) =\mathfrak{e} \max\left(0, 1-2\gamma\right)  \;, 
	\end{array} \right.
	\end{equation} 
	where, $E_{\kappa,\gamma}^{\star}$ represents the energy value maximizing $\ergo^{(1)}(D_{\kappa,\gamma}; {\mathfrak{e}})$ in the interval $[0,E]$. In the derivation above we used the fact that $\fixedergo(D_{\kappa,\gamma};\mathfrak{e})$ in this parameter region is a monotonically increasing convex function of $\mathfrak{e}$, so $\chi(D_{\kappa,\gamma},\mathfrak{e})$ is provided by the linear interpolation between its final value (i.e. $\fixedergo^{(1)}(D_{\kappa,\gamma}; 1)$) and initial value  (i.e. $\fixedergo^{(1)}(D_{\kappa,\gamma}; 0)=0$).
For $\kappa < \tfrac{\gamma}{1-\gamma}$ instead, $\fixedergo^{(1)}(D_{\kappa,\gamma}; {\mathfrak{e}})$ is concave and can achieve its maximum inside the allowed domain~$[0,1]$. 
A study of the first order derivative of the function allows us to verify that, if $\kappa \geq  4\gamma-1$ (a region which incidentally fully includes also the cases $\kappa \geq \tfrac{\gamma}{1-\gamma}$ analyzed before), the maximum is still achieved for ${\mathfrak{e}}=1$, meaning that $\fixedergo^{(1)}(D_{\kappa,\gamma}; {\mathfrak{e}})$ remains monotonous w.r.t. to ${\mathfrak{e}}$. 
In this case we can hence write
	\begin{equation} 
\forall \kappa \in[4\gamma -1,\tfrac{\gamma}{1-
	\gamma}] \; \Longrightarrow \; \left\{
	\begin{array}{l}
	E_{\kappa,\gamma}^{\star}({\mathfrak{e}}) = {\mathfrak{e}} \;, \\ \\ 
	\ergo^{(1)}(D_{\kappa,\gamma}; {\mathfrak{e}})  = \fixedergo^{(1)}(D_{\kappa,\gamma}; {\mathfrak{e}})\;, \\ \\
	\chi(D_{\kappa,\gamma},\mathfrak{e}) = \ergo^{(1)}(D_{\kappa,\gamma}; \mathfrak{e})\;, 
	\end{array} \right.
	\end{equation} 
where now the value of $\chi(D_{\kappa,\gamma},\mathfrak{e})$ coincides with 
$\ergo^{(1)}(D_{\kappa,\gamma}; \mathfrak{e})$ since the latter is a concave function.
Finally, for  $\kappa \leq 4 \gamma -1$ the maximum of $\fixedergo^{(1)}(D_{\kappa,\gamma}; {\mathfrak{e}})$ is achieved for 
		\begin{eqnarray}
		E_{\kappa,\gamma}:=\tfrac{-\kappa + \sqrt{\gamma-\kappa(1-\gamma)}}{2(\gamma-\kappa)}\;, \end{eqnarray} 
		implying 
		\begin{equation} 
\kappa \in [0,4\gamma -1] 
	\; \Longrightarrow \; \left\{
	\begin{array}{l}
	E_{\kappa,\gamma}^{\star}({\mathfrak{e}}) = \min\{ {\mathfrak{e}},E_{\kappa,\gamma}\} 
	 \;, \\ \\ 
	\ergo^{(1)}(D_{\kappa,\gamma}; {\mathfrak{e}}) = \fixedergo^{(1)}(D_{\kappa,\gamma}; E_{\kappa,\gamma}^{\star}({\mathfrak{e}}))\;, \\ \\
	\chi(D_{\kappa,\gamma},\mathfrak{e}) = \ergo^{(1)}(D_{\kappa,\gamma}; \mathfrak{e})\;, 
	\end{array} \right.
	\end{equation} 
	 where, to express $\chi(D_{\kappa,\gamma},\mathfrak{e})$, we used the fact that since $\fixedergo^{(1)}(D_{\kappa,\gamma}; {\mathfrak{e}})$ is concave then 
 $\ergo^{(1)}(D_{\kappa,\gamma}; \mathfrak{e})$ is concave too.
In Fig.~\ref{fig: ADC_DEPH1_maxErg} we report a plot of $\ergo^{(1)}(D_{\kappa,\gamma}; \mathfrak{e})$
  and $E_{\kappa,\gamma}^{\star}(\mathfrak{e})$ in terms of $\gamma$ and $\kappa$ for $\mathfrak{e}=1$.
 
Invoking~(\ref{exact}) we can hence compute the value of the local (and separable-input local) ergotropic capacitance of the model,
		\begin{equation} 
C_{\rm loc} \left(D_{\kappa,\gamma}; \mathfrak{e} \right)=
	\;  \left\{
	\begin{array}{l}
	\mathfrak{e} \max\left(0, 1-2\gamma\right) 
	\quad   \forall \kappa  \in [ \tfrac{\gamma}{1-
	\gamma},1] 
	 \;, \\ \\ 
	 \fixedergo^{(1)}(D_{\kappa,\gamma}; \mathfrak{e}) 
	 \quad    \forall \kappa \in[4\gamma -1,\tfrac{\gamma}{1-
	\gamma}]\;,  \\ \\
	\fixedergo^{(1)}(D_{\kappa,\gamma}; \min\{ \mathfrak{e},E_{\kappa,\gamma}\})  \\ 
	 \qquad   \qquad  \qquad  \qquad\forall \kappa \in [0,4\gamma -1] \;.
	\end{array} \right. \label{EXCLOC} 
	\end{equation} 
We stress that, at variance with what happens with the GADCs, in this case the evaluation of $C_{\rm sep} \left(D_{\kappa,\gamma}, \mathfrak{e} \right)$ cannot be performed due to the fact that property (\ref{step1}) no longer applies. Regarding the MAWERs, invoking Theorem~\ref{theo1}, from~(\ref{EXCLOC})  
we can write

	\begin{figure}[]
		\setlength{\lineskip}{3pt}
		\centering
		\includegraphics[width=0.9\linewidth]{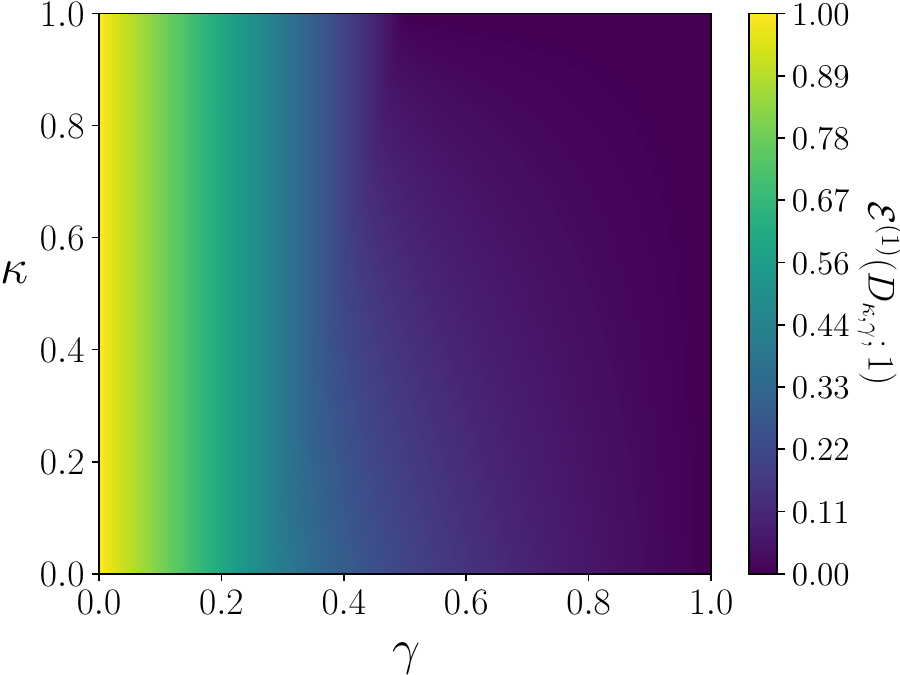}
		\includegraphics[width=0.9\linewidth]{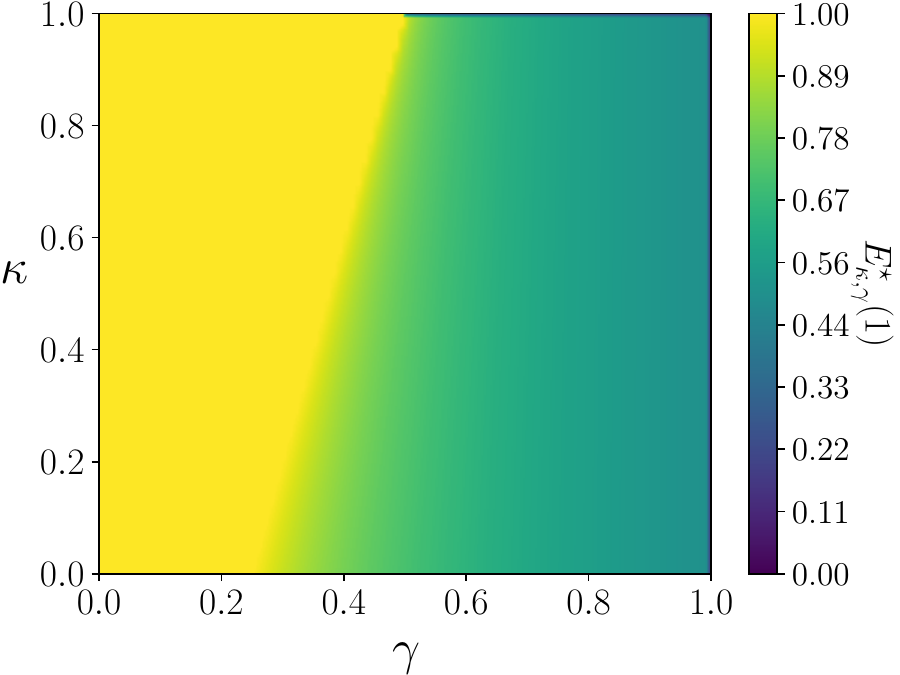}
		\caption{ Plot of $\ergo^{(1)}(D_{\kappa,\gamma};1)$ (\textbf{Upper panel}) and of the optimal input energy $E_{\kappa,\gamma}^{\star}(\mathfrak{e})$ (\textbf{Lower panel}) as functions of $\gamma$ and $\kappa$ for $\mathfrak{e}=1$.}
		\label{fig: ADC_DEPH1_maxErg}
	\end{figure}
\begin{equation} \label{eq:mawerdeph}
{\cal J}_{\rm loc} \left(D_{\kappa,\gamma} \right)=
	\;  \left\{
	\begin{array}{l}
	 \max\left(0, 1-2\gamma\right) 
	\quad   \forall \kappa  \in [ \tfrac{\gamma}{1-
	\gamma},1] 
	 \;, \\ \\ 
	 (1-\gamma)(1-\kappa) \qquad  \forall \kappa \in[0,\tfrac{\gamma}{1-
	\gamma}]\;,
	\end{array} \right. 
\end{equation} 
where we used the fact that 
	\begin{eqnarray}
	\left. \frac{\partial  \fixedergo^{(1)}(D_{\kappa,\gamma}; \mathfrak{e})}{\partial \mathfrak{e}}
	\right|_{\mathfrak{e}=0} = (1-\gamma)(1-\kappa)\;. 
	\end{eqnarray} 
It is worth stressing that in the strong  dephasing regime (i.e. for $\kappa  \in [ \tfrac{\gamma}{1-\gamma},1]$) the value of ${\cal J}_{\rm loc}\left(D_{\kappa,\gamma} \right)$ coincides with the value of the local MAWER we could get in the absence of dephasing (i.e. for the ADC $\adc$) when using ``classical'' input states $|\phi_E^{(n)}\rangle$ (see Eq.~(\ref{ergoloc_gadc})). In Fig.~\ref{fig: ADC_DEPH__Diff_asympt_Erg1} we plot the difference of the expressions in Eq. (\ref{eq:mawerdeph}), we see that coherence is needed to obtain higher values of the local MAWER in the regime of weak dephasing, i.e. for $\kappa\in[0,\tfrac{\gamma}{1-\gamma}]$. This difference is a proper indicator of the quantum advantage, because the first expression in~\eqref{eq:mawerdeph} gives the extractable energy if the initial state is classical, while the second shows us the extractable energy if the input state possess quantum coherence.
	
	\begin{figure}[h!]
		\centering
		\includegraphics[width = 1\linewidth]{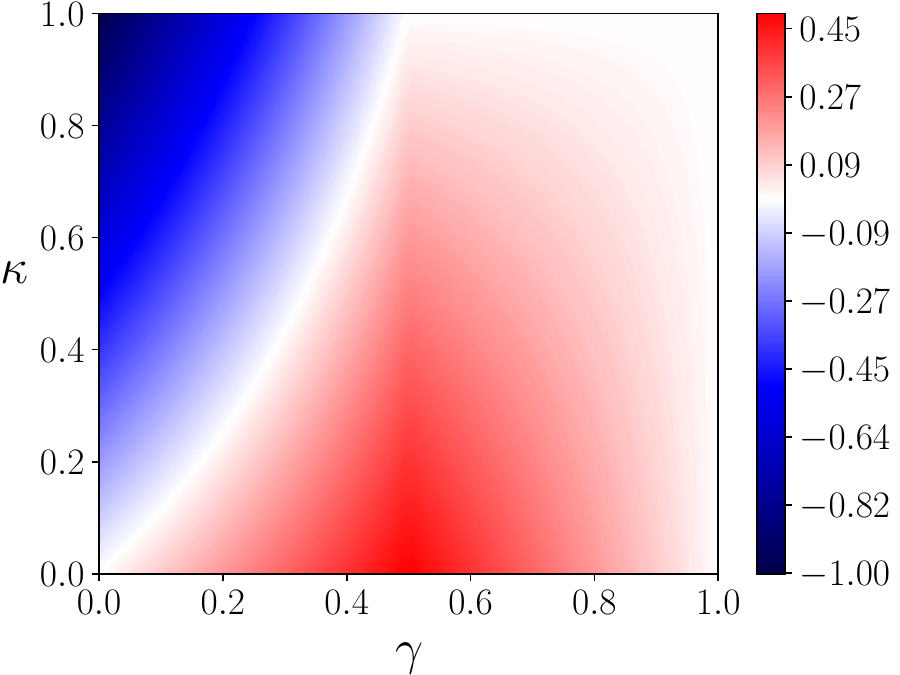}
		\caption{Difference between the two possible values of $\mathcal{J}_{\rm loc}(D_{\kappa,\gamma})$ in (\ref{eq:mawerdeph}), we see that for $\kappa \in(\tfrac{\gamma}{1-\gamma},1]$ the classical strategy with input state $\ket{\phi^{(n)}_E}$ is optimal, while for $\kappa \in[0,\tfrac{\gamma}{1-\gamma}]$ the coherent strategy $\ket{\psi_E^{(n)}}$ outperforms the other. So, in a vast region of parameters, coherence enhances the efficiency of the energy storage. If the difference is positive, then $\mathcal{J}_{\rm loc}(D_{\kappa,\gamma})$ is attained by a coherent input state; while, if it is negative, the local MAWER is obtained via a classical input state. So the this difference clearly represent the parameter region where we obtain a quantum advantage.}
		\label{fig: ADC_DEPH__Diff_asympt_Erg1}
	\end{figure}

\section{Discussion} \label{sec:discuss}
In this article we examined to what extent a QB composed of an arbitrary number of q-cells can maintain its initial energy after the action of environmental noise. In particular, we focused on the self-discharging process of the battery as a result of contact with a thermal environment and we studied how the loss of coherence due to extra dephasing noise affects the stability of a QB. We did so by adopting the framework of the quantum work capacitances and MAWERs. 
More explicitly, among other considerations:
\begin{itemize}
    \item In Theorem \ref{theo1} we show how the capacitance and the MAWER associated with a specific functional (e.g. the ergotropy) are directly and formally related.
    
    \item In Sec. \ref{sec:adc}, addressing the ADC, we provide an analytical expression for the maximizer input energy of $\ergo^{(1)}$ and $
\fixedergo^{(1)}$, which are consequently also  found analytically. We  show how for the ADC $C_{\rm loc} = C_{\rm loc, sep} = \chi = 
  \ergo^{(1)}$ and provide upper and lower bounds for $C_{\cal E}$. Remarkably, concerning the MAWER, we prove that ${\cal J}_{\cal E} = {\cal J}_{\rm sep} ={\cal J}_{\rm loc} ={\cal J}_{\rm loc, sep} $ and give their exact value.
  
 \item In Sec. \ref{sec:GADC}, addressing the GADC, as for the ADC we find the analytical values for the input energy maximizer of $\ergo^{(1)}$ and $\fixedergo^{(1)}$, the analytical expression of $C_{\rm loc} = C_{\rm loc, sep} = \chi = C_{\rm sep} = 
  \ergo^{(1)}$ and upper and lower bounds for $C_{\cal E}$.
  For MAWERs we show that
  ${\cal J}_{\cal E} = {\cal J}_{\rm sep}={\cal J}_{\rm loc} ={\cal J}_{\rm loc, sep} $ and these have an analytical expression.

\item In Sec. \ref{sec:metaphysics} we analyzed the simultaneous action of dephasing and ADC. We found, for various regimes of the channel parameters, analytical expressions for the input energy maximizer of $\ergo^{(1)}$ and $\fixedergo^{(1)}$, the values of $\chi$, $C_{\rm loc}$ and $\cal{J}_{\rm loc}$.
  
\end{itemize} 
Interestingly, in the case of the generalized amplitude damping channel, the analysis shows that quantum coherence can enhance the efficiency of energy storage, while this advantage is lost in the presence of sufficiently large dephasing. Remarkably, in the special case of a QB in contact with a very high temperature environment (i.e. for GADCs in the limit $\eta \to \frac{1}{2}$), the battery is able to preserve all of its input ergotropy. \\

Building on these considerations and results, this work may serve as a starting point for a more realistic analysis of the performance of quantum batteries. While noisy channels as the ones analyzed in this paper are in general simplified models of those occurring in experimental setups, our findings can be of assistance in experiment design by providing indications of the level of noise resistance and storing strategies required to observe a quantum advantage for energy storage quantum devices.  \\

\begin{acknowledgments}
We acknowledge financial support by MUR (Ministero dell’ Universit\`a e della Ricerca) through the following projects: 
PNRR MUR project PE0000023-NQSTI, PRIN 2017 Taming complexity via Quantum Strategies: a Hybrid Integrated Photonic approach (QUSHIP) Id. 2017SRN-BRK, and project PRO3 Quantum Pathfinder. S.C. is also supported by a grant through the IBM-Illinois
\end{acknowledgments}

\clearpage

\appendix
\section{Ergotropy: a brief review} \label{sec.rev} 

The maximization on the right-hand-side of Eq.~(\ref{ergoeeeDEFinvariance}) admits a closed expression in terms of the passive counterpart  of $\dstate$, i.e. the
 density matrix $\dstate_{\text{pass}}$ obtained  by operating on the latter via 
 a unitary rotation that transforms its eigenvectors $\{ |\lambda_\ell\rangle\}_{\ell}$ into the eingenvectors $\{ |E_\ell\rangle\}_{\ell}$ of the system Hamiltonian $\ham$, matching the corresponding eingenvalues in reverse order, i.e. 
 \begin{eqnarray} 
 \left.
 \begin{array}{l} 
 \dstate =\sum_{\ell =1}^{d} \lambda_\ell |\lambda_\ell\rangle \! \langle\lambda_\ell|  \\ \\
\ham=\sum_{\ell =1}^{d} E_\ell |E_\ell\rangle\!\langle E_\ell| 
 \end{array} \right\} \mapsto
 \dstate_{\text{pass}} := \sum_{\ell =1}^{d} \lambda_\ell |E_\ell\rangle\!\langle E_\ell| \;, \nonumber \\ \label{deftutto} 
 \end{eqnarray} 
 where for all $\ell=1,\cdots, d-1$,  we set $\lambda_{\ell} \geq \lambda_{\ell+1}$ and 
 $E_\ell  \leq E_{\ell+1}$. Explicitly we can write 
 \begin{equation}\label{ergoeee} 
	\ergo(\dstate;\ham) =  \en(\dstate;\ham)-  \en(\dstate_{\text{pass}};\ham)  
	=\en(\dstate;\ham)- \sum_{\ell=1}^{d} \lambda_\ell E_\ell\;.
	\end{equation}
	Notice that for the special case in which the quantum system is a qubit with Hamiltonian $\hat{H}=|1\rangle\langle 1|$,
	Eq.~(\ref{ergoeee}) gives,
	\begin{eqnarray}\label{ergoqubit}  
	\ergo(\dstate;\ham) = \langle 1 |\dstate|1\rangle - \lambda_{\min}(\dstate)\;,  
	\end{eqnarray} 
	with $\lambda_{\min}(\dstate)$ the minimum eigenvalue of $\dstate$. 
	
The total-ergotropy $\ergo_{\rm tot}(\dstate;\ham)$ is a regularized version of  
$\ergo(\dstate;\ham)$ that emerges when considering scenarios where one has at disposal an arbitrary large number of identical copies of the input state $\dstate$. 
Formally it is defined as 
\begin{eqnarray}\label{totergoeeeDEFtotal} 
\ergo_{\rm tot}(\dstate;\ham)  &:=& \lim_{N\rightarrow \infty} \frac{\ergo(\dstate^{\otimes N};\ham^{(N)})}{N}\;,
	\end{eqnarray}
where for fixed $N$ integer, $\ham^{(N)}$ is the total Hamiltonian of the $N$ copies of the system obtained by assigning to each of them the same $\ham$ (no interactions being included). 
One can show\cite{Alicki2013, PhysRevA.105.012414} that the limit in Eq.~ (\ref{totergoeeeDEFtotal}) exists and corresponds to the maximum of $\frac{\ergo(\dstate^{\otimes N};\ham^{(N)})}{N}$ with respect to all possible $N$, implying in particular that $\ergo_{\rm tot}(\dstate;\ham)$ is at least as large as $\ergo(\dstate;\ham)$, i.e.
 \begin{eqnarray} \label{qubitimpo1} 
 \ergo_{\rm tot}(\dstate;\ham) =\sup_{N\geq 1} \frac{\ergo(\dstate^{\otimes N};\ham^{(N)})}{N}  \geq \ergo(\dstate;\ham)\;. 
 \end{eqnarray} 
 The  case where the system has dimension 2 represents an exception to this rule as in this case one has that
  \begin{eqnarray} \label{qubitsimpo} 
 \ergo_{\rm tot}(\dstate;\ham) =\frac{\ergo(\dstate^{\otimes N};\ham^{(N)})}{N}  =\ergo(\dstate;\ham)\;, 
 \end{eqnarray} 
 for all inputs and for all $N$. 
Most notably $\ergo_{\rm tot}(\dstate;\ham)$ can be  expressed via a single letter formula that mimics Eq.~(\ref{ergoeee}), i.e. 
\begin{equation} \label{GIBBS} 
\ergo_{\rm tot}(\dstate;\ham) =
 \en(\dstate;\ham)-  \en^{({\beta_\star})}_{{\small{\rm GIBBS}}}(\ham)\;, 
\end{equation} 
where for $\beta \geq 0$ and $Z_{\beta}({\ham}) : =\mbox{Tr}[ e^{ - \beta \ham}] = 
 \sum_{\ell=1}^d  e^{-\beta E_\ell}$
\begin{eqnarray} 
\en^{({\beta})}_{{\small{\rm GIBBS}}}(\ham)&:=& \en(\hat{\tau}_{\beta};\ham)=
- \frac{d}{d\beta} \ln Z_{\beta}({\ham}) 
\;, 
\end{eqnarray} 
is the mean energy of the thermal Gibbs state 
\begin{eqnarray} \label{GIBBS1} 
\hat{\tau}_{\beta}:= e^{ - \beta \ham}/Z_{\beta}({\ham})\;, \end{eqnarray} 
with effective inverse temperature $\beta$; while finally $\beta_\star$ is chosen so that  $\dstate_{\beta_\star}$
has the same von Neumann entropy of $\dstate$, i.e.
\begin{eqnarray} S_{\beta_\star} = S(\dstate) := -\mbox{Tr}[\dstate \ln \dstate]\;,
\end{eqnarray}  
with $S_{\beta}:= -\mbox{Tr}[\hat{\tau}_{\beta} \ln \hat{\tau}_{\beta}]= - \beta  \frac{d}{d\beta} \ln Z_{\beta}({\ham})+ \ln Z_\beta({\ham})$.

For many-body quantum systems, local ergotropy is defined by restricting to only local transformations the optimization over the unitary set of Eq.~(\ref{ergoeeeDEFinvariance}), i.e.
\begin{equation}\label{ergoeeeDEFinvarianceloc} 
\ergo_{\rm loc}(\dstate;\ham)  := \max_{\hat{U} \in {\mathbb U}_{\rm loc}(d)}\Big\{  \en(\dstate;\ham) 
-\en(\hat{U} \dstate \hat{U}^\dag;\ham)\Big\}   \;.
\end{equation}
By construction it provides a lower bound for $\ergo(\dstate;\ham)$ and in case the system is represented by $k$ non-interacting particles (i.e. if $\ham$ is given by a sum of local terms $\ham=\ham_1+ \cdots + \ham_k$) it reduces to the sum of local 
contributions,  
\begin{eqnarray}\label{ergoeeeDEFinvariancelocFACT} 
\ergo_{\rm loc}(\dstate;\ham)  = \sum_{i=1}^k \ergo(\dstate_k;\ham_k)\;,
\end{eqnarray}
where for $i=\in\{ 1,\cdots, k\}$ the $\dstate_i$ represents the reduced density matrix of the $i$-th subsystem of the model.

\section{Alternative characterization of the MAWER} \label{app:physmaw}
First of all we notice that for any quantum channel $\Lambda$ the quantity $\ergo_{\diamond}^{(n)}(\Lambda;E)$ is non-decreasing in $n$ for any fixed input energy $E \in (0,\mathfrak{e}_{max}]$. I.e. $\forall\; n\in\mathbb{N}$ $\ergo^{(n+1)}_{\diamond}(\Lambda;E) \geq \ergo^{(n)}_{\diamond}(\Lambda;E)$, this property is proven in \cite{quantumworkcapacitances}. Due to this fact it easily follows that
\begin{equation} \label{eq:suplim}
\sup_{n\geq\lceil E/\mathfrak{e}_{max} \rceil} \ergo^{(n)}_{\diamond}(\Lambda;E) = \lim_{n\to\infty}\ergo^{(n)}_{\diamond}(\Lambda;E) \; .
\end{equation}
We now define the quantity $\tilde{\mathcal{J}}_{\diamond}(\Lambda)$ as
\begin{equation}
\tilde{\mathcal{J}}_{\diamond}(\Lambda) := \liminf_{E\to\infty}\lim_{n\to\infty}\frac{\ergo^{(n)}_{\diamond}(\Lambda;E)}{E} \; ,
\end{equation}
clearly for any quantum channel $\Lambda$ $\tilde{\mathcal{J}}_{\diamond}(\Lambda) \leq \mathcal{J}_{\diamond}(\Lambda)$ we now need to prove the converse.
\begin{lemma} \label{lemma:physmawer}
For any quantum channel $\Lambda$ the following equation holds true: 
\begin{equation}
\tilde{\mathcal{J}}_{\diamond}(\Lambda) \geq \sup_{\mathfrak{e}\in(0,\mathfrak{e}_{max}]}\frac{C_{\diamond}(\Lambda;\mathfrak{e})}{\mathfrak{e}} = \mathcal{J}_{\diamond}(\Lambda) \;.
\end{equation}
\begin{proof}
By definition we have that for any $\mathfrak{e}\in(0,\mathfrak{e}_{max}]$
\begin{eqnarray}
\tilde{\mathcal{J}}_{\diamond}(\Lambda) &\geq& \liminf_{E\to\infty}\frac{\ergo^{(\lceil E/\mathfrak{e} \rceil - 1)}_{\diamond}(\Lambda;(\lceil E/\mathfrak{e} \rceil -1)\mathfrak{e})}{E} \nonumber \\
&=&  \liminf_{E\rightarrow\infty} 
\left(\tfrac{\lceil E/\mathfrak{e}\rceil -1}{E} \right)
\frac{\ergo_{\diamond}^{(\lceil E/\mathfrak{e}\rceil-1)}(\Lambda;(\lceil E/\mathfrak{e}\rceil -1)\mathfrak{e})}{\lceil E/\mathfrak{e}\rceil-1} \nonumber \\
&=& \frac{C_{\diamond}(\Lambda;\mathfrak{e})}{\mathfrak{e}} \; ,
\end{eqnarray}
in the last passage we used the fact that for all the terms of the expression the limit inferior is equal to the limit, so we conclude 
\end{proof}
\end{lemma}
Thanks to Lemma \ref{lemma:physmawer} we can now state that the MAWER can be written as
\begin{equation}
\mathcal{J}_{\diamond}(\Lambda) = \lim_{E\to\infty}\lim_{n\to\infty}\frac{\ergo^{(n)}_{\diamond}(\Lambda;E)}{E} \; ,
\end{equation}
here the limits do not commute.

\end{document}